\documentclass[11pt]{article}

\usepackage{amsmath,amssymb,amsthm,amsfonts,latexsym,bbm,xspace,graphicx,float,mathtools}
\usepackage[backref,colorlinks,citecolor=blue,bookmarks=true]{hyperref}
\usepackage[letterpaper,margin=1in]{geometry}
\usepackage{color}
\usepackage{algorithm}
\usepackage{algorithmic}

\usepackage{tweaklist}
\renewcommand{\enumhook}{\setlength{\topsep}{2pt}%
  \setlength{\itemsep}{-2pt}}
\renewcommand{\itemhook}{\setlength{\topsep}{2pt}%
  \setlength{\itemsep}{-2pt}}

\newtheorem{theorem}{Theorem}
\newtheorem{informal}[theorem]{Informal Theorem}
\newtheorem{corollary}{Corollary}
\newtheorem{lemma}{Lemma}

\newtheorem{proposition}{Proposition}

\newtheorem{definition}{Definition}

\newtheorem{observation}{Observation}

\newtheorem{remark}{Remark}

\newenvironment{prevproof}[2]{\noindent {\em {Proof of {#1}~\ref{#2}:}}}{$\Box$\vskip \belowdisplayskip}

\newcommand{\junk}[1]{}

\junk{

}

\newcommand{\poly}{{\rm poly}}

\newcommand{\costasnote}[1]{{\color{red}{#1}}}
\newcommand{\costasfnote}[1]{{\color{red}{\footnote{\costasnote{{\bf Costas says:}~#1}}}}}
\newcommand{\mattnote}[1]{{\color{black}{#1}}}

\newcommand{\notshow}[1]{{}}

\def \obj {{\cal O}}

\DeclareMathOperator{\argmax}{argmax}

\definecolor{MyGray}{rgb}{0.8,0.8,0.8}

\begin{document}
\title{Bayesian Truthful {\em Mechanisms} for Job Scheduling\\ from Bi-criterion Approximation {\em Algorithms}}
\author {
Constantinos Daskalakis\thanks{Supported by a Sloan Foundation Fellowship, a Microsoft Research Faculty Fellowship and NSF Award CCF-0953960 (CAREER) and CCF-1101491.}\\
EECS, MIT \\
\tt{costis@mit.edu}
\and
S. Matthew Weinberg\thanks{Supported by a Microsoft Research Graduate Fellowship, NSF Graduate Research Fellowship and NSF award CCF-1101491.}\\
EECS, MIT\\
\tt{smw79@mit.edu}
}

\addtocounter{page}{-1}
\maketitle
\begin{abstract}

We provide polynomial-time approximately optimal Bayesian mechanisms for makespan minimization on unrelated machines as well as for max-min fair allocations of indivisible goods, with approximation factors of $2$ and $\min\{m-k+1, \tilde{O}(\sqrt{k})\}$ respectively, matching the approximation ratios of best known polynomial-time \emph{algorithms} (for max-min fairness, the latter claim is true for certain ratios of the number of goods $m$ to people $k$). Our mechanisms are obtained by establishing a polynomial-time approximation-sensitive reduction from the problem of designing approximately optimal {\em mechanisms} for some arbitrary objective ${\cal O}$ to that of designing bi-criterion approximation {\em algorithms} for the same objective ${\cal O}$ plus a linear allocation cost term. Our reduction is itself enabled by extending the  celebrated ``equivalence of separation and optimization''~\cite{GrotschelLS81,KarpP80} to also accommodate bi-criterion approximations. Moreover, to apply the reduction to the specific problems of makespan and max-min fairness we develop polynomial-time bi-criterion approximation algorithms for makespan minimization with costs and max-min fairness with costs, adapting the algorithms of \cite{ShmoysT93},~\cite{BezakovaD05} and~\cite{AsadpourS07} to the type of bi-criterion approximation that is required by the reduction. 

\end{abstract}
\thispagestyle{empty}

\newpage

\section{Introduction} \label{sec:intro}

\vspace{-3pt}Job shop scheduling is a fundamental problem that has been intensively studied in operations research and computer science in several different flavors. The specific one that we consider in this paper, called {\em scheduling unrelated machines}, pertains to the allocation of $m$ indivisible jobs for execution to $k$ machines so as to minimize the time needed for the last job to be completed, called the {\em makespan} of the schedule. The input is the processing time $p_{ij}$ of each machine $i$ for each job $j$. The problem is ${\rm NP}$-hard to $(3/2-\epsilon)$-approximate, for any $\epsilon >0$, but a polynomial-time $2$-approximation algorithm is known~\cite{LenstraST87}. An overview of algorithmic work on this problem~can~be~found~in~\cite{Hochbaum96}.

\smallskip Starting in the seminal work of Nisan and Ronen~\cite{NisanR99}, scheduling unrelated machines has also become paradigmatic for investigating the relation between the complexity of mechanism and algorithm design. Mechanism design can be viewed as the task of optimizing an objective over ``strategic inputs.'' In comparison to algorithm design where the inputs are known, in mechanism design the inputs are owned by rational agents who must be incentivized to share enough information about their input so that the desired objective can be optimized. The question raised by~\cite{NisanR99} is how much this extra challenge degrades our ability to optimize objectives:

\medskip \noindent ~~\begin{minipage}[h]{15.8cm}
{\em How much more difficult is mechanism design for a certain objective compared to algorithm design for that same objective?}
\end{minipage}

\medskip \noindent In the context of scheduling unrelated machines, suppose that the machines are rational agents who know their own processing times for the jobs, but want to minimize the sum of processing times of the jobs assigned to them minus the payment made to them by the mechanism. If the machines are rational, is it still possible to (approximately) minimize makespan? 

\smallskip Indeed, there are two questions pertaining to the relation of algorithm and mechanism design that are important to answer. The first is comparing the performance of the optimal mechanism to that of the optimal algorithm. In our setting, the question is whether there are mechanisms whose makespan is (approximately) optimal with respect to the real $p_{ij}$'s, which (at least a priori) are only known to the machines. Nisan and Ronen show that the classical VCG mechanism achieves a factor $k$ approximation to the optimal makespan~\cite{NisanR99}, but since their work no constant factor approximation has been obtained. We overview known upper and lower bounds in Section~\ref{sec:relatedwork}.

 The second question pertaining to the relation of algorithm and mechanism design is of computational nature. The question is whether polynomial-time (approximately) optimal  mechanisms exist  for objectives for which polynomial-time (approximately) optimal algorithms exist. In our context, there exist polynomial-time algorithms whose makespan is approximately optimal with respect to the optimal makespan of any feasible schedule~\cite{LenstraST87}, so the question is whether there exist polynomial-time mechanisms whose makespan is approximately optimal with respect to that of any mechanism. This is the question that we study in this paper.

Before proceeding it is worth mentioning that (outside of makespan minimization) this question has been intensively studied,  and the results are discouraging. In particular, a sequence of recent results~\cite{PapadimitriouSS08,BuchfuhrerDFKMPSSU10,Dobzinski11,DobzinskiV12} have identified welfare maximization problems for which polynomial-time constant factor approximation algorithms exist, but where no polynomial-time mechanism is better than a polynomial-factor approximation, subject to well-believed complexity theoretic assumptions. At the same time, we have also witnessed a recent surge in the study of mechanisms in Bayesian settings, where the participants of the mechanism (in our case machines) have types (in our case processing times for jobs) drawn from a prior distribution that is common knowledge. The existence of priors has been shown~\cite{HartlineL10,HartlineKM11,BeiH11} to sidestep several intractability results including the ones for welfare maximization referenced above. 
In view of this experience, it is natural to ask:

\vspace{5pt} \noindent ~~\begin{minipage}[h]{15.8cm}
{\em Are there approximately optimal, computationally efficient mechanisms for makespan minimization in Bayesian settings?}
\end{minipage}



 We provide a positive answer to this question, namely (see Theorem~\ref{thm:truthfulmakespan} for a formal statement)
\begin{informal}\label{informal thm: makespan mechanisms}
In Bayesian settings, there is a polynomial-time {2}-approximately optimal mechanism for makespan minimization for unrelated machines.
\end{informal}
\noindent In particular,  the approximation factor achieved by our mechanism exactly matches the best known approximation factor achieved by polynomial-time algorithms~\cite{LenstraST87}. In fact, our proof establishes a polynomial-time, approximation sensitive, black-box reduction from the problem of designing a mechanism for makespan minimization to the problem of designing a bi-criterion approximation for the generalized assignment problem~\cite{ShmoysT93,ShmoysT93b}. We explain our reduction and the type of bi-criterion approximation that is required in Section~\ref{sec:black-box intro}. We discuss prior work on mechanisms for makespan minimization in Section~\ref{sec:relatedwork}, noting here that the best known approximation factors prior to our work were polynomial, in general.


\medskip A problem related to makespan minimization is that of {\em max-min fair allocation of indivisible goods}, abbreviated to {\em max-min fairness}. In the language of job scheduling, this can be described as looking for an assignment of jobs to machines that maximizes the minimum load---rather than minimizing the maximum load, which is the goal in makespan minimization. While the two problems are related, the best known polynomial-time approximation algorithms for max-min fairness achieve factors that are polynomial in the number of jobs or machines. We overview algorithmic work on the problem in Section~\ref{sec:relatedwork}, noting here that there are several, mutually undominated approximation algorithms, whose approximation guarantees have different dependences on the number of jobs $m$, machines $k$, and other parameters of the problem.
Our contribution here is to obtain polynomial-time Bayesian mechanisms matching the approximation factor of some of those algorithms, namely (see Theorem~\ref{thm:truthfulfairness} for a formal statement)
\begin{informal}\label{informal thm: fairness results}
In Bayesian settings, there are polynomial-time approximately-optimal mechanisms for max-min fairness whose fairness guarantees are respectively:
\begin{enumerate}
\item within a factor of $\tilde{O}(\sqrt{k})$ of ${\rm OPT}$;
\item within a factor of $O(m-k+1)$ of ${\rm OPT}$.
\end{enumerate}
where ${\rm OPT}$ is the fairness achieved by the optimal mechanism.
\end{informal}
\noindent In particular, our approximation guarantees match those of the approximation algorithms provided by~\cite{AsadpourS07} and 
\cite{BezakovaD05}, which both lie on the Pareto boundary of what is achievable by polynomial-time algorithms. Our contribution here, too, can be viewed as pushing mechanism design up to speed with algorithm design for the important objective of max-min fairness. Our proof is enabled by a polynomial-time, approximation sensitive, black-box reduction from mechanism design for max-min fairness to bi-criterion approximation algorithm design {\em for max-min fairness with allocation costs}, for which we recover approximation guarantees matching those of~\cite{AsadpourS07,BezakovaD05} in Section~\ref{sec:makespan}.

Our mechanism to algorithm reduction, enabling Theorems~\ref{informal thm: makespan mechanisms} and~\ref{informal thm: fairness results} is discussed next.

\subsection{Black-Box Reductions in Mechanism Design} \label{sec:black-box intro}
 A natural approach towards Theorems~\ref{informal thm: makespan mechanisms} and~\ref{informal thm: fairness results} is establishing a polynomial-time reduction from (approximately) optimizing over mechanisms to (approximately) optimizing over algorithms. This approach has already been shown fruitful for welfare maximization. Indeed, recent work establishes such a reduction for welfare maximization in Bayesian settings~\cite{HartlineL10,HartlineKM11,BeiH11}. Roughly speaking, it is shown that black-box access to an $\alpha$-approximation algorithm for an arbitrary welfare maximization problem can be leveraged to obtain an $\alpha$-approximately optimal mechanism for the same welfare maximization problem. 

Unfortunately, recent work has ruled out such black-box reduction for makespan minimization~\cite{ChawlaIL12}. This impossibility result motivated recent work by the authors, where it is shown that adding a linear allocation cost term to the algorithmic objective can bypass this impossibility~\cite{CaiDW12,CaiDW13,CaiDW13b}. Specifically, it is shown in~\cite{CaiDW13b} that finding an ($\alpha$-approximately) optimal mechanism for an arbitrary objective ${\cal O}$ can be reduced to polynomially many black-box calls to an ($\alpha$-approximately) optimal algorithm for the same objective ${\cal O}$, perturbed by an additive allocation cost term.\footnote{Technically, their result holds for maximization objectives, but our work here provides the necessary modifications for minimization objectives, as well as  the important generalization to $(\alpha,\beta)$-approximations discussed below.} This reduction was used to find polynomial-time (approximately) optimal mechanisms for the important objective of revenue~\cite{CaiDW12,CaiDW13} as well as non-linear objectives such as max-min fairness for divisible goods~\cite{CaiDW13b}. 

On the other hand, adding a (possibly negative) allocation cost term may turn an objective~${\cal O}$ that can be (approximately) optimized in polynomial-time into one that cannot be optimized to within any finite factor. This is precisely what happens if we try to carry out the reduction of~\cite{CaiDW13b} for makespan minimization or max-min fairness with indivisible goods. More precisely:
\begin{itemize}
\item To find a polynomial-time $\alpha$-approximately optimal mechanism for makespan minimization, the reduction of~\cite{CaiDW13b} requires a polynomial-time $\alpha$-approximately optimal algorithm for the problem of {\em scheduling unrelated machines with costs}. This is similar to {scheduling unrelated machines}, except that now it also costs $c_{ij}$ (which may be positive, negative, or $0$) to assign job $j$ to machine $i$, and we are looking for an allocation $\vec{x} \in \{0,1\}^{km}$ of jobs to machines that minimizes
\begin{align}M(\vec{x}) + \sum_{ij}c_{ij} x_{ij},\label{eq:funny objective}\end{align}
where $M(\vec{x})=\max_i \sum_j p_{ij} x_{ij}$ is the makespan of the allocation $\vec{x}$. In words, we want to find a schedule that minimizes the sum of makespan and cost of the allocation. However, it is easy to see that it is NP-hard to optimize~\eqref{eq:funny objective} to within any finite factor even when restricted to instances whose optimum is guaranteed to be positive.\footnote{This can be seen via a simple modification of an inapproximability result given in~\cite{LenstraST87}. For the problem of scheduling unrelated machines, they construct instances with integer-valued makespan that is always $\ge 3$ and such that it is NP-hard to decide whether the makespan is $3$ or $\ge 4$. We can modify their instances to scheduling unrelated machines with costs instances by giving each job a cost of $\frac{z-3}{2n+m}$ on every machine for an arbitrary $z>0$. Then the total cost of any feasible solution is exactly $z-3$. So their proof immediately shows that it is NP-hard to determine if these instances have optimal makespan + cost that is $z$ or $\ge 1+z$. Since $z$ was arbitrary, this shows that no finite approximation factor is possible.} 
\item Similarly, to find a polynomial-time $\alpha$-approximately optimal mechanism for max-min fairness, the reduction of~\cite{CaiDW13b} requires a polynomial-time $\alpha$-approximately optimal algorithm for the problem of {\em max-min fairness with allocation costs.} In the notation of the previous bullet, we are looking for an allocation $\vec{x} \in \{0,1\}^{km}$ of jobs to machines that maximizes
\begin{align}{F}(\vec{x}) + \sum_{ij}c_{ij} x_{ij},\label{eq:funny objective fairness}\end{align}
where ${F}(\vec{x}) = \min_i \sum_j p_{ij} x_{ij}$ is the load of the least loaded machine under allocation $\vec{x}$. Again, it is easy to see that it is NP-hard to optimize~\eqref{eq:funny objective fairness} to within any finite factor.\footnote{Indeed, Bezakova and Dani~\cite{BezakovaD05} present a family of max-min fairness instances such that it is {\rm NP}-hard to distinguish between ${\rm OPT}\ge 2$ and ${\rm OPT} \le 1$. To each of these instances add a special machine and a special job such that the processing-time and cost of the special machine for the special job are $2$ and $-1$ respectively, while the processing-time and cost of the special machine for any non-special job or of any non-special machine for the special job are $0$ and $0$ respectively. Also, assign $0$ cost to any non-special machine non-special job pair. In the resulting max-min fairness with costs instances it is ${\rm NP}$-hard to distinguish between ${\rm OPT} \ge 1$ and ${\rm OPT}=0$, hence no finite approximation is possible.}
\end{itemize}

\subsubsection{A single-criterion to bi-criterion approximation-sensitive reduction} 
The inapproximability results identified above motivate us to develop a novel reduction that is more robust to adding the allocation cost term to the mechanism design objective.
We expect that our new reduction will reach a much broader family of mechanism design objectives, and indeed as a corollary of our new reduction we obtain Theorems~\ref{informal thm: makespan mechanisms} and~\ref{informal thm: fairness results} for the important objectives of makespan and max-min fairness, where the reduction of~\cite{CaiDW13b} fails. 

Our new approach is based on the concept of $(\alpha,\beta)$-approximation of objectives modified by allocation costs, defined in Section~\ref{sec:prelims}. Instead of presenting the concept in full generality here, let us describe it in the context of the makespan minimization objective and its resulting scheduling unrelated machines with costs problem. For $\beta \le 1 \le \alpha$, we will say that an allocation $\vec{x} \in \{0,1\}^{km}$ of jobs to machines is an $(\alpha,\beta)$-approximation to a scheduling unrelated machines with costs instance  iff
\begin{align}\beta \cdot M(\vec{x}) + \sum_{ij}c_{ij} x_{ij} \le \alpha \cdot \min_{\vec{x}' \in \{0,1\}^{km}}\left(M(\vec{x}') + \sum_{ij}c_{ij} x'_{ij}\right);\label{eq:alpha-beta}\end{align}
that is, we discount the makespan term in the objective, before comparing to the optimum. 

Setting $\beta=1$ in~\eqref{eq:alpha-beta} recovers the familiar notion of $\alpha$-approximation, but taking $\beta <1$ might make the problem easier. Indeed, we argued earlier that it is NP-hard to achieve any finite $\alpha$ when $\beta=1$. On the other hand, we can exploit the bi-criterion result of Shmoys and Tardos for the generalized assignment problem~\cite{ShmoysT93} to get a polynomial-time algorithm achieving {$\beta={1 \over 2}$} and $\alpha=1$~\cite{ShmoysT93}. The proof of the following proposition is presented in Section~\ref{sec:makespan}.
\begin{proposition} \label{inf thm:alpha beta for makespan with costs} There is a polynomial-time {$(1,{1 \over 2})$}-approximation algorithm for scheduling unrelated machines with (possibly negative) costs.
\end{proposition}

Given such $(\alpha,\beta)$-approximation algorithms for objectives modified by allocation cost, we show how to obtain $({\alpha \over \beta})$-approximately optimal mechanisms, by establishing an appropriate mechanism to algorithm reduction described informally below. 

\begin{informal} \label{inf thm: alpha-beta}There is a generic, polynomial-time reduction from the problem of computing an $({\alpha \over \beta})$-approximately optimal polynomial-time mechanism for some arbitrary objective $\cal O$ (under arbitrary feasibility constraints and allowable bidder types) to the problem of $(\alpha,\beta)$-approximately optimizing that same objective $\cal O$ modified by virtual welfare (under the same constraints and allowable bidder types). Whenever the allowable bidder types are additive in the mechanism design instance,  the algorithmic objective becomes $\cal O$ plus a linear allocation cost term.
\end{informal}
\noindent See Theorem~\ref{thm:objective} in Section~\ref{sec:BMeD and approximation} for a formal statement. The main technical challenge in establishing our reduction is extending the celebrated ``equivalence of separation and optimization''~\cite{GrotschelLS81,KarpP80} to also accommodate $(\alpha,\beta)$-approximations---see Theorem~\ref{thm:LPalphabeta}. Theorem~\ref{informal thm: makespan mechanisms} is obtained by combining  Proposition~\ref{inf thm:alpha beta for makespan with costs} and Theorem~\ref{inf thm: alpha-beta}. 

\medskip To apply our mechanism to algorithm reduction to max-min fairness, we need $(\alpha,\beta)$-approximation algorithms for max-min fairness with allocation costs. Since we have a maximization objective, we are now looking to compute allocations $\vec{x} \in \{0,1\}^{km}$ such that
\begin{align}\beta \cdot {F}(\vec{x}) + \sum_{ij}c_{ij} x_{ij} \ge \alpha \cdot \max_{\vec{x}' \in \{0,1\}^{km}}\left({F}(\vec{x}') + \sum_{ij}c_{ij} x'_{ij}\right),\label{eq:alpha-beta fairness}\end{align}
for some $\beta \ge 1 \ge \alpha$. In this case, we are allowed to boost the fairness part of the objective before comparing to $\alpha \cdot {\rm OPT}$. Again, even though no finite $\alpha$ is achievable in polynomial time when $\beta=1$, we can adapt the algorithms of~\cite{AsadpourS07,BezakovaD05} to obtain finite $(\alpha, \beta)$-approximation algorithms for max-min fairness with costs.

\begin{proposition} \label{inf thm:alpha beta for max min fairness with costs} There is a polynomial-time {$(\frac{1}{2},\tilde{O}(\sqrt{k}))$}-approximation algorithm as well as a $(1,O(m-k+1))$-approximation algorithm for max-min fairness with costs.
\end{proposition}
 \noindent The proof of Proposition~\ref{inf thm:alpha beta for max min fairness with costs} is given in Section~\ref{sec:makespan}. It extends the algorithms of~\cite{AsadpourS07,BezakovaD05} to the presence of allocation costs, showing that the natural linear programming relaxation can be rounded so that the cost term does not increase, while the fairness term decreases by a factor of $\tilde{O}(\sqrt{k})$ or $O(m-k+1)$ respectively. Combining Proposition~\ref{inf thm:alpha beta for max min fairness with costs} with Theorem~\ref{inf thm: alpha-beta} gives Theorem~\ref{informal thm: fairness results}.
\subsection{Related Work}\label{sec:relatedwork}

Makespan minimization and max-min fairness have been extensively studied in both algorithm and mechanism design settings and numerous different models. It would be impossible to survey all related literature, but we highlight some results most related to ours below.

\paragraph{Makespan.} A long line of work following the seminal paper of Nisan and Ronen~\cite{NisanR99} addresses the question of ``how much better can the optimal makespan be when compared to the optimal makespan obtained by a truthful mechanism?'' The same paper showed that the answer is at most a factor of $k$, and also that the answer is at least $2$ for deterministic, dominant strategy truthful, prior-free mechanisms.\footnote{A mechanism is called ``prior-free'' if it does not make any distributional assumption about the processing times of the machines.} It was later shown that the answer is at least $1+\phi$ (the golden ratio) as $k\rightarrow \infty$~\cite{ChristodoulouKV07,KoutsoupiasV07}, and that the answer is in fact $k$ for the restricted class of anonymous mechanisms~\cite{AshlagiDL12}. It is conjectured that the answer is indeed $k$ for all deterministic, prior-free mechanisms. Similar (but slightly different) bounds are known for the same question with respect to randomized prior-free mechanisms~\cite{NisanR99,MualemS07,ChristodoulouKK07,LuY08a,LuY08b,Lu09}. More recently, the same question has been studied for prior-independent (rather than prior-free) mechanisms~\cite{ChawlaHMS13}. Prior-independent mechanisms make distributional assumptions about the processing times, but do not use the specifics of the distributions, just their properties. In particular, when the processing times are drawn from a machine-symmetric product distribution with Monotone Hazard Rate (MHR) marginals, Chawla, Hartline, Malec and Sivan show that the answer is at most a factor of $O(m/k)$, and at most a factor of $O(\sqrt{\log k})$, when all processing times are i.i.d.~\cite{ChawlaHMS13}. Without the MHR assumption, they obtain bicriterion results.\footnote{Specifically, they obtain the same bounds with respect to a different benchmark, namely the optimal expected makespan using only a fraction of the $k$ machines.} The question has also been studied in \emph{related machines} settings, where each job has a size (public) and each machine has a speed (private). Due to the single-dimensional nature of the problem, the answer is now exactly a factor of $1$~\cite{ArcherT01}. Thus, focus has  shifted towards the same question for computationally efficient truthful mechanisms, and constant-factor approximations~\cite{ArcherT01} and PTAS's~\cite{DhangwatnotaiDDR08,ChristodoulouK10} are known. 



The focus of our work is different than most previous works, in that we do not study the gap between the algorithmic and the mechanismic optimum. Instead, our focus is computational, aiming for (approximately) optimal and computationally efficient mechanisms, regardless of how their performance compares to the performance of optimal algorithms. Still, prior work has already made some progress on this problem: we know that the VCG mechanism is a computationally efficient $k$-approximation~\cite{NisanR99}, and that the mechanisms of~\cite{ChawlaHMS13} provide an approximation ratio of $O(m/k)$ when the prior is a machine-symmetric product distribution with MHR marginals, and a ratio of $O(\sqrt{\log k})$ if additionally the marginals are i.i.d. In addition, the NP-hardness result of~\cite{LenstraST87} implies that our problem is also NP-hard to approximate better than $3/2-\epsilon$, for any $\epsilon>0$. On this front, our work greatly improves  the state-of-the-art as we give the first constant-factor approximations for unrestricted settings. (The guarantees of~\cite{ChawlaHMS13} are constant in settings where $m = O(k)$ or $k = O(1)$ and the prior is a machine-symmetric product distribution with MHR marginals.) Indeed, our approximation guarantee (factor of $2$) matches that of the best polynomial-time algorithm for makespan minimization~\cite{LenstraST87}.

\paragraph{Max-Min Fairness.} Fair division has been studied extensively in Mathematics and Economics for over 60 years; see e.g.~\cite{Steinhaus48,Knaster46,BramsT96}. Several different flavors of the problem have been studied: divisible or indivisible goods (in our terminology: jobs), with or without monetary transfers to the players (in our terminology: machines), and several different notions of fairness. For the allocation of indivisible goods, virtually all mechanisms proposed in the literature do not optimize max-min fairness, aiming instead at other fairness guarantees (such as envy-freeness or proportionality), very commonly trade off value from received items by the players with monetary transfers to or from the players in the fairness guarantee, and are often susceptible to strategic manipulations. For max-min fairness, Bezakova and Dani~\cite{BezakovaD05} propose a prior-free mechanism for $2$ players, which guarantees half of the optimal fairness, albeit under restrictions on the strategies that the players can use. They also show that max-min fairness cannot be optimally implemented truthfully in prior-free settings. In fact, Mualem and Schapira show that no truthful deterministic prior-free mechanism can obtain any approximation to the optimal max-min fairness~\cite{MualemS07}. 

Perhaps the poor state-ot-the-art on mechanisms for max-min fairness owes to the fact that, already as an algorithmic problem (i.e. even when the players' true values for the items are assumed exactly known), max-min fairness has proven quite challenging, indeed significantly more so than makespan minimization. In particular, all state-of-the-art algorithms only provide polynomial approximation guarantees~\cite{BezakovaD05,AsadpourS07, KhotP07, BateniCG09, ChakrabartyCK09}, while the best known computational hardness result is just a factor of $2$~\cite{BezakovaD05}. Specifically, the guarantees lying on the Pareto boundary of what is achievable in polynomial time for the unrestricted problem are approximation factors of $m-k+1$ (\cite{BezakovaD05}), $\tilde{O}(\sqrt{k})$ (\cite{AsadpourS07}), and $O(m^{1/\epsilon})$ (\cite{BateniCG09, ChakrabartyCK09}). Due to this, a \emph{restricted} version of the problem is often studied, where every job has a fixed processing time $p_j$, and every machine is either capable or incapable of processing each job (that is, $p_{ij} \in \{p_j,0\}$). For the restricted version, the state of the art is an $O(\log \log k /\log \log \log k)$-approximation due to Bansal and Sviridenko~\cite{BansalS06}. Asadpour, Feige, and Saberi~\cite{AsadpourFS08} also proved that the integrality gap of the configuration LP used in~\cite{BansalS06} has an integrality gap of $5$, and provided a heuristic (but not polynomial-time) rounding algorithm. $O(1)$-approximations were obtained by Bateni, Charikar and Guruswami and independently by Chakrabarty, Chuzhoy and Khanna by further restricting the graph structure of which machines can process which jobs (i.e. by limiting the number of machines that can process each specific job or requiring that this graph be acyclic)~\cite{BateniCG09, ChakrabartyCK09}.

In view of this literature, our results provide the first approximately optimal mechanisms for max-min fairness. Indeed, our approximation factors match those of approximation algorithms on the Pareto boundary of what is achievable in polynomial time~\cite{BezakovaD05,AsadpourS07}. So, in particular, our mechanisms cannot be strictly improved without progress on the algorithmic front. {Obtaining these mechanisms is already quite involved (see Section~\ref{sec:makespan} and Appendix~\ref{app:makespanfairness}), and we leave open for future investigation the problem of matching the bounds obtained in~\cite{BateniCG09,ChakrabartyCK09} for the general problem,~\cite{AsadpourFS08, BansalS06} for the restricted problem, and~\cite{BateniCG09, ChakrabartyCK09} for the futher restricted version.}

\paragraph{Black-Box Reductions in Mechanism Design, and Inapproximability.} We have already reviewed several results studying black-box reductions from mechanism to algorithm design. In fact, the classical VCG mechanism can already be viewed as a mechanism to algorithm reduction for the important objective of welfare, and Myerson's celebrated mechanism~\cite{Myerson81} as a mechanism to algorithm design reduction, indeed a special case of that in~\cite{CaiDW13b}, for revenue. We also identified in Section~\ref{sec:black-box intro} mechanism design objectives for which the reduction results in inapproximable algorithmic problems. Indeed, there is a precedence of this phenomenon explored in~\cite{HaghpanahIMM11, BateniHSZ13} for the simpler problem of revenue maximization in {\em single-dimensional} settings. In the settings studied by these papers, the algorithmic problem resulting from the reduction (namely virtual welfare maximization) is highly inapproximable. Nevertheless, they side-step this intractability by exploiting the fact that the algorithmic problem need only be solved well in an average-case sense (in particular, in expectation over the bidder's virtual values) rather than on an instance-to-instance basis, as well as the fact that, in single-dimensional settings, the virtual values have very well-understood structure. This allows for the design of polynomial-time algorithms that obtain a reasonable approximation guarantee on average, while possibly performing poorly on some instances. While this approach is fruitful in single-dimensional settings and the revenue objective, we expect that considerably more effort is required in order to apply it to multidimensional settings or non-revenue objectives, due to the lack of structural understanding of virtual values in these settings. Indeed, the revenue maximization hardness of approximation result of~\cite{CaiDW13b} displays that beyond simple single-dimensional settings, even average-case approximation algorithms for virtual welfare can be computationally infeasible.

\notshow{\subsubsection{Black Box Reductions in Bayesian Mechanism Design}\label{sec:relatedBlackBox}

We have already reviewed prior work on black box reductions from mechanism to algorithm design in Bayesian settings, describing results for maximizing welfare~\cite{HartlineL10,HartlineKM11,BeiH11,ChawlaIL12}, revenue~\cite{CaiDW12,CaiDW12b,CaiDW13,CaiDW13b}, and other objectives~\cite{ChawlaIL12,CaiDW13b}. The results closer to our work are those in~\cite{CaiDW13,CaiDW13b}, which are obtained by establishing approximation preserving versions of the celebrated ``equivalence of separation and optimization''~\cite{GrotschelLS81,KarpP80}. Within this context, our main contribution is a further extension of the equivalence of separation and optimization to accommodate $(\alpha,\beta)$-approximations (see Theorem~\ref{thm:LPalphabeta}), which immediately allows our mechanism to algorithm reduction framework (Theorem~\ref{thm:objective}) to accommodate $(\alpha,\beta)$-approximations to virtual welfare modified objectives, instead of only traditional approximation algorithms as in~\cite{CaiDW13b}. Many important problems (for instance, minimizing makespan or maximizing max-min fairness) are NP-hard to approximate within any finite factor once virtual welfare is added to the objective, but $(O(1),O(1))$-approximations may be tractable (as indeed is the case for makespan). Our results on makespan are important in their own right, but also serve to provide evidence that this extended framework may allow for the development of interesting approximately optimal mechanisms for problems that were previously intractable.}

\section{Preliminaries} \label{sec:prelims}

Our formal setup is a special case of that in~\cite{CaiDW13b}. We repeat it here for completeness. Throughout the preliminaries and entire paper, we state our definitions and results when the goal is to minimize an objective (such as makespan). Everything extends to maximization objectives (such as fairness) with the obvious changes (switching $\leq$ to $\geq$, $\min$ to $\max$, etc.). We often note~the~required~changes.

\paragraph{Mechanism Design Setting.} The mechanism designer has a set of feasible outcomes $\mathcal{F}$ to choose from. Each bidder participating in the mechanism may have several possible {\em types}. A bidder's {type} determines a value for each possible outcome in $\mathcal{F}$. Specifically, a bidder's type induces a function $t: \mathcal{F} \rightarrow\mathbb{R}$. $T_i$ denotes the set of all possible types of bidder $i$, which we assume to be finite. Bidder $i$'s type is drawn from a distribution $\mathcal{D}_i$ over $T_i$, which is known to the designer and all other bidders. Bidder $i$ knows his own type. Bidders are {\em quasi-linear} and {\em risk-neutral}. That is, the utility of a bidder of type $t$ for a randomized outcome (distribution over outcomes) $X \in \Delta(\mathcal{F})$, when he is charged (a possibly random price with expectation) $p$, is $\mathbb{E}_{x \leftarrow X}[t(x)] - p$. Therefore, we may extend $t$ to take as input distributions over outcomes as well, with $t(X) = \mathbb{E}_{x \leftarrow X}[t(x)]$. A {\em type profile} $\vec{t}=(t_1,\ldots,t_k)$ is a collection of types for each bidder. We assume that the types of the bidders are independent so that $\mathcal{D} = \times_i \mathcal{D}_i$ is the distribution over type profiles. 

\paragraph{Mechanisms.} A (direct) mechanism consists of two functions, a (possibly randomized) allocation rule and a (possibly randomized) price rule, and we allow these rules to be correlated. The allocation rule takes as input a type profile $\vec{t}$ and (possibly randomly) outputs an allocation $A(\vec{t}) \in \mathcal{F}$. The price rule takes as input a profile $\vec{t}$ and (possibly randomly) outputs a price vector $P(\vec{t})$. A direct mechanism invites bidders to report their type to the mechanism, and the bidders may or may not  report their type truthfully. When the profile $\vec{t}$ is reported to the mechanism $M = (A,P)$, the (possibly random) allocation $A(\vec{t})$ is selected and each bidder $i$ is charged the (possibly random) price $P_i(\vec{t})$. In the definitions below, we discuss the \emph{interim allocation rule} of a mechanism. This is a function that takes as input a bidder $i$ and a type $t_i \in T_i$ and outputs the distribution of allocations that bidder $i$ sees when reporting type $t_i$ over the randomness of the mechanism and the other bidders' types, if they report truthfully. Specifically, if the interim allocation rule of $M=(A,P)$ is $X$, then $X_i(t_i)$ is a distribution satisfying $$\Pr[x\leftarrow X_i(t_i) ] = \mathbb{E}_{\vec{t}_{-i} \leftarrow \mathcal{D}_{-i}}\left[\Pr[A(t_i;\vec{t}_{-i}) = x]\right],$$
where $\vec{t}_{-i}$ is the vector of types of all bidders but bidder $i$ in $\vec{t}$, and ${\cal D}_{-i}$ is the distribution~of~$\vec{t}_{-i}$. {Similarly, the {\em interim price rule} of the mechanism maps some bidder $i$ and type $t_i \in T_i$ of that bidder to the expected price bidder $i$ sees when reporting $t_i$, i.e. $p_i(t_i)=\mathbb{E}\left[\mathbb{E}_{\vec{t}_{-i} \leftarrow \mathcal{D}_{-i}}[P_i(\vec{t})]\right].$}

With these definitions, a mechanism is said to be {\em Bayesian Incentive Compatible (BIC)} if it is in every bidder's best interest to truthfully report their type to the mechanism, if all other bidders truthfully report their type. That is, if {$X,p$ are the interim allocation and price rules} of the mechanism, then $t_i(X_i(t_i)) - p_i(t_i) \geq t_i(X_i(t'_i)) - p_i(t'_i)$ for all $i$ and $t_i,t'_i \in T_i$. \mattnote{A mechanism is $\epsilon$-BIC if each bidder can gain at most an additive $\epsilon$ by misreporting their type. That is, $t_i(X_i(t_i)) - p_i(t_i) \geq t_i(X_i(t'_i)) - p_i(t'_i) - \epsilon$. In order to make this additive guarantee meaningful, we assume that all types have been normalized so that $t_i(X) \in [0,1]$ for all $X\in \mathcal{F}$.} A mechanism is said to be {\em Individually Rational (IR)} if it is in every bidder's best interest to participate in the mechanism, no matter their type. That is, $t_i(X_i(t_i)) - p_i(t_i) \geq 0$, for all $i$, $t_i \in T_i$. \mattnote{A mechanism is said to be \emph{ex-post IR} if furthermore every bidder receives non-negative utility by telling the truth for every realization of the other agents' types and the randomness in the mechanism (and not just in expectation).}

\paragraph{Goal of the designer.} The designer's goal is to design a BIC and IR mechanism that minimizes (or maximizes) the expected value of some objective function, $\obj$, when encountering a bidder profile sampled from $\mathcal{D}$ and the bidders report truthfully, which is in their interest to do if the mechanism is BIC. For simplicity of notation, we restrict our attention in this paper to objective functions $\obj$ that take as input a type profile $\vec{t}$ and a randomized outcome $X \in \Delta(\mathcal{F})$, and output a quantity $\obj(\vec{t},X) \in \mathbb{R}_+$. \mattnote{We say that $\mathcal{O}$ is $b$-bounded if whenever $t_i(X) \in [0,1]$ for all $i$, $\obj(\vec{t},X) \in [0,b]$. Because our results accommodate an additive $\epsilon$ error, we will restrict attention only to $\mathcal{O}$ that are $\poly(k)$-bounded. Note that makespan and fairness are both $1$-bounded, and welfare is $k$-bounded, so this is not a restrictive assumption.}

 We also restrict attention to objective functions $\obj$ such that $\obj(\vec{t},X) = \mathbb{E}_{x \leftarrow X}[\obj(\vec{t},x)]$. In other words, $\obj$ is really just a function of types and outcomes in $\mathcal{F}$, and is extended to randomized outcomes in $\Delta({\cal F})$ by taking expectations (makespan and fairness are examples of such an objective).\footnote{This is a special case of the setting considered in~\cite{CaiDW13b}. We restrict our attention to this case just to simplify notation, because more generality is not needed for makespan and fairness. However, Theorem~\ref{thm:objective} applies to objectives that also depend on the prices charged, as well as objectives that are sensitive to randomness in non-linear ways (but still must be concave in distributions over outcomes/prices for maximization objectives and convex for minimization objectives). Note that ``deterministic objectives'' such as makespan, fairness, welfare, and revenue that are extended to randomized outcomes by taking expectation behave linearly with respect to randomness and are therefore both concave and convex in the sense that is relevant for the theorem.}

\paragraph{Formal Problem Statements.} We define the computational problems \textbf{B}ayesian \textbf{Me}chanism \textbf{D}esign (BMeD) and \textbf{G}eneralized \textbf{O}bjective \textbf{O}ptimization \textbf{P}roblem (GOOP), which played a central role in~\cite{CaiDW13b} and will also play a central role in this paper.\footnote{These problems were named MDMDP (Multi-Dimensional Mechanism Design Problem) and 2-SADP (Solve-Any Differences Problem) in~\cite{CaiDW13b}. We change their names here for a more accurate description of the problems.} We state both BMeD and GOOP for minimization objectives, but both problems are also well-defined for maximization objectives with the obvious modifications, discussed below. In the following definitions, we denote by ${\cal V}$ a set of types (i.e. functions mapping ${\cal F}$ to $\mathbb{R}$), and by ${\cal V}^{\times}$ the closure of ${\cal V}$ under~addition~and~scalar~multiplication.

\smallskip \textbf{BMeD($\mathcal{F}$,$\mathcal{V}$,$\obj$):} {\sc Input:} For each bidder $i \in [k]$, a finite set of types $T_i \subseteq {\cal V}$ and a distribution ${\cal D}_i$ over $T_i$. {\sc Goal:} Find a feasible (outputs an outcome in $\mathcal{F}$ with probability $1$), BIC, and IR mechanism~$M$ that minimizes $\obj$ in expectation, when $k$ bidders with types sampled from $\mathcal{D} = \times_i \mathcal{D}_i$ play $M$ truthfully, where the minimization is with respect to all feasible, BIC, and IR mechanisms. $M$ is said to be an $\alpha$-approximation to BMeD if the expected value of $\obj$ is at most $\alpha$ times the optimal one.\footnote{\mattnote{By ``find a mechanism'' we formally mean ``output a computational device that will take as input a profile of types $\vec{t}$ and output (possibly randomly) an outcome in $\mathcal{F}$ and a price to charge each agent.'' The runtime of this device is of course relevant, and will be addressed in our formal theorem statements.}}

\textbf{GOOP($\mathcal{F}$, $\mathcal{V}$, $\mathcal{O}$):}  {\sc Input:}  $f \in \mathcal{V}^{\times}$, $g_i \in \mathcal{V}$ ($1\leq i \leq k$), and a multiplier $w \geq 0$. {\sc Goal:}  find a feasible (possibly randomized) outcome $X \in \Delta(\mathcal{F})$ such that:
\begin{align*}
\left(w \cdot \obj((g_1,\ldots,g_k),X)\right) + f(X) = \min_{X' \in \mathcal{F}} \left\{\left(w \cdot \obj((g_1,\ldots,g_k),X')\right) + f(X')\right\}.
\end{align*}

We define a bi-criterion notion of approximation for GOOP, saying that $X$ is an {\em $(\alpha,\beta)$-approximation} to GOOP for some $\beta \le 1 \le \alpha$ iff:
\begin{align}
\beta \left(w \cdot \obj((g_1,\ldots,g_k),X)\right) + f(X) \leq \alpha \left( \min_{X' \in \mathcal{F}} \left\{\left(w \cdot \obj((g_1,\ldots,g_k),X')\right) + f(X')\right\}\right). \label{eq:alpha-beta approximation}
\end{align}
An $X$ satisfying~\eqref{eq:alpha-beta approximation} with $\beta=1$ is an {\em $\alpha$-approximation} to GOOP, which is the familiar notion of approximation for minimization problems. If $\beta <1$, our task becomes easier as the contribution of the ${\cal O}$ part to the objective we are looking to minimize is discounted. 

Within the context of our reduction from BMeD to GOOP, as well as that of previous reductions, one should interpret the function $f$ in the GOOP instance as representing virtual welfare, and each $g_i$ as representing the type reported by agent $i$.

\begin{remark} For maximization objectives $\cal O$, we replace $\min$ by $\max$ in the definition of GOOP, and we invert the direction of the inequality in~\eqref{eq:alpha-beta approximation}. Moreover, the feasible range of parameters for an $(\alpha,\beta)$-approximation are now $\alpha \le 1 \le \beta$.
\end{remark}

\paragraph{Scheduling Unrelated Machines with Costs.} There are $k$ machines and $m$ indivisible jobs. Each machine~$i$ can process job $j$ in time $p_{ij} \ge 0$. Additionally, processing job $j$ on machine $i$ costs $c_{ij}$ units of currency, where $c_{ij}$ is unrestricted and in particular could be negative. An assignment of jobs to machines is a $km$-dimensional vector $\vec{x}$ such that $x_{ij} \in \{0,1\}$, for all $i$ and $j$, where job $j$ is assigned to machine $i$ iff $x_{ij} = 1$.  We denote by $M(\vec{x}) = \max_{i} \{\sum_j x_{ij} p_{ij}\}$ the makespan of an assignment, by $F(\vec{x}) = \min_i \{\sum_j x_{ij}p_{ij}\}$ the fairness of an assignment, and by $C(\vec{x}) = \sum_i \sum_j x_{ij} c_{ij}$ the cost of an assignment. In the makespan minimization problem, an assignment is valid iff $\sum_i x_{ij} = 1$ for all jobs $j$, while in the fairness maximization problem, an assignment is valid iff $\sum_i x_{ij} \le 1$. To avoid carrying these constraints around, we use the convention that $M(\vec{x}) = \infty$, if $\sum_i x_{ij} \neq 1$ for some $j$, and $F(\vec{x}) = -\infty$, if $\sum_i x_{ij} > 1$ for some $j$. It will also be useful for analysis purposes to consider fractional assignments of jobs, which relax the constraints $x_{ij} \in \{0,1\}$ to $x_{ij} \in [0,1]$. This corresponds to assigning an $x_{ij}$-fraction of job $j$ to machine $i$ for all pairs $(i,j)$. Notice that $M(\vec{x})$, $F(\vec{x})$ and $C(\vec{x})$ are still well-defined for fractional assignments.

The goal of makespan minimization with costs is to find an assignment $\vec{x} \in \{0,1\}^{km}$ satisfying $M(\vec{x}) + C(\vec{x}) = \min_{\vec{x}' \in \{0,1\}^{km}}\{M(\vec{x}') + C(\vec{x}')\}$. In the language of GOOP, this is GOOP($\{0,1\}^{km}$, $\{$additive functions with non-negative coefficients$\}$, $M$). \mattnote{The processing times $\vec{p}_i$ correspond to the functions $g_i$ that are input to GOOP, and the costs $\vec{c}$ corresponds to the function $f$.} For $\alpha \ge 1 \ge \beta$, an $(\alpha,\beta)$-approximation for this problem is an assignment $\vec{x} \in \{0,1\}^{km}$ with $\beta M(\vec{x}) + C(\vec{x}) \leq \alpha \min_{\vec{x}' \in \{0,1\}^{km}}\{M(\vec{x}') + C(\vec{x}')\}$. 

The goal of fairness maximization with costs is to find an assignment $\vec{x} \in \{0,1\}^{km}$ satisfying $F(\vec{x}) + C(\vec{x}) = \max_{\vec{x}' \in \{0,1\}^{km}}\{F(\vec{x}') + C(\vec{x}')\}$. In the language of GOOP, this is GOOP($\{0,1\}^{km}$, $\{$additive functions with non-negative coefficients$\}$, $F$). \mattnote{Again, the processing times $\vec{p}_i$ correspond to the functions $g_i$ that are input to GOOP, and the costs $\vec{c}$ corresponds to the function $f$.} For $\alpha \le 1 \le \beta$, an $(\alpha,\beta)$-approximation for this problem is an assignment $\vec{x} \in \{0,1\}^{km}$ with $\beta F(\vec{x}) + C(\vec{x}) \geq \alpha \max_{\vec{x}' \in \{0,1\}^{km}}\{F(\vec{x}') + C(\vec{x}')\}$. 

Note that in the case of maximizing fairness, sometimes the jobs are thought of as gifts and the machines are thought of as children (and the problem is called the Santa Claus problem). In this case, it makes sense to think of the children as having \mattnote{value for the gifts (and preferring more value to less value)} instead of the machines having \mattnote{processing time for jobs (and preferring less processing time to more)}. For ease of exposition, we will stick to the jobs/machines interpretation, although our results extend to the gifts/children interpretation as well.

\paragraph{Implicit Forms.} For any feasible mechanism $M = (A, P)$ for a $BMeD(\mathcal{F},\mathcal{V},\mathcal{O})$ instance, we define (as in~\cite{CaiDW13b}) the three components of its implicit form $\vec{\pi}_I^M = (O^M, \vec{\pi}^M, \vec{p}^M)$ as follows. 
\begin{itemize}
\item $O^M = \mathbb{E}_{\vec{t} \leftarrow \mathcal{D}}[O(\vec{t}, A(\vec{t}) )]$. That is, $O^M$ is the expected value of $\mathcal{O}$ when agents sampled from $\mathcal{D}$ play mechanism $M$.
\item For all agents $i$ and types $t, t' \in T_i$, $\pi^M_i(t, t') = \mathbb{E}_{\vec{t}_{-i} \leftarrow \mathcal{D}_{-i}}[t( A ( t',\vec{t}_{-i} ) )]$. That is, $\pi^M_i(t, t')$ is the expected value of agent $i$ with real type $t$ from reporting type $t'$ to the mechanism $M$. The expectation is taken over any randomness in $M$ as well as the other agents' types, assuming they are sampled from $\mathcal{D}_{-i}$. 
\item For all agents $i$ and types $t \in T_i$, $p^M_i(t) = \mathbb{E}_{\vec{t}_{-i} \leftarrow \mathcal{D}_{-i}} [P_i(t, \vec{t}_{-i})]$. That is, $p^M_i(t)$ is the expected price paid by agent $i$ when reporting type $t$ to the mechanism $M$. The expectation is taken over any randomness in $M$ as well as the other agents' types, assuming they are sampled from $\mathcal{D}_{-i}$. 
\end{itemize}

We can also talk about implicit forms separately from mechanisms, and call any $(1+\sum_i |T_i|^2 + \sum_i |T_i|)$-dimensional vector an implicit form. We say that an implicit form $\vec{\pi}_I = (O, \vec{\pi}, \vec{p})$ is \emph{feasible} for a specific $BMeD(\mathcal{F},\mathcal{V},\mathcal{O})$ instance if there exists a feasible mechanism $M$ for that instance such that $O \geq O^M$, $\vec{\pi} = \vec{\pi}^M$, and $\vec{p} = \vec{p}^M$. We say that the mechanism $M$ \emph{implements} the implicit form $\vec{\pi}_I$ when these inequalities hold. For maximization objectives, we instead constrain $O \leq O^M$.\footnote{The relaxations $O \geq O^M$, for minimization, and $O \leq O^M$, for maximization objectives, instead of $O = O^M$, is required for technical reasons.} We denote by $F(\mathcal{F},\mathcal{D},\mathcal{O})$ the set of all feasible implicit forms (with respect to a specific instance of $BMeD(\mathcal{F},\mathcal{V},\mathcal{O})$). It is shown in~\cite{CaiDW13b} that $F(\mathcal{F},\mathcal{D},\mathcal{O})$ is a convex set.

\notshow{
\paragraph{Representation Questions.} Notice that both MDMDP and 2-SADP are parameterized by ${\cal F}$ and ${\cal V}$. As we aim to leave these sets unrestricted, we assume that their elements are represented in a computationally meaningful way. That is, we assume that elements of $\mathcal{F}$ can be indexed using $O(\log |{\cal F}|)$ bits {and are input to functions that evaluate them via this representation. We assume that elements $f \in \mathcal{V}$ are input to functions either via a description of a Turing machine that evaluates $f$ (and the size of the description counts towards the size of the input), or as a black box. All our reductions presented in this paper apply whether or not the input functions are given explicitly or as a black box.\footnote{{When we claim that we can solve problem $\mathcal{P}_1$ given black-box access to a solution to problem $\mathcal{P}_2$, we mean that the functions input to problem $\mathcal{P}_1$ may be given either explicitly or as a black box, and that they are input in the same form to $\mathcal{P}_2$.}}} Whenever we evaluate the running time of an algorithm for either MDMDP or 2-SADP, or of a reduction from one problem to another, we count the time spent in evaluating a function input to these problems as one.}

\section{BMeD and Approximation} \label{sec:BMeD and approximation}
The goal of this section is to prove the following theorem, generalizing Theorem~4 of~\cite{CaiDW13b} to accommodate both $(\alpha,\beta)$-approximations and minimization objectives. Informally, Theorem~\ref{thm:objective} states that there is a polynomial-time black box reduction from $(\frac{\alpha}{\beta})$-approximating BMeD($\mathcal{F}$,$\mathcal{V}$,$\mathcal{O}$) to $(\alpha,\beta)$-approximating GOOP($\mathcal{F}$,$\mathcal{V}$,$\mathcal{O}$). Theorem~4 of~\cite{CaiDW13b} is a special case of our theorem here for maximization objectives and $(\alpha,1)$-approximations. Throughout this section, we will use $b$ to represent an upper bound on $\max_i\{|T_i|\}$ and the bit complexity of $\obj(\vec{t},X)$, $\Pr[t_i]$, and $t_i(X)$, for all $i, t_i \in T_i, X \in \mathcal{F}$ for the given $BMeD(\mathcal{F},\mathcal{V},\mathcal{O})$ instance. We will also use the notation $rt_G(x)$ to denote an upper bound on the running time of algorithm $G$ on input of bit complexity $x$. 


\begin{theorem}\label{thm:objective}
Let $G$ be an $(\alpha,\beta)$-approximation algorithm for GOOP($\mathcal{F}$,$\mathcal{V}$,$\mathcal{O}$), for some $\alpha \ge 1 \ge \beta>0$, and some minimization objective $\cal O$. Also, fix any $\epsilon>0$. Then there is an approximation algorithm for BMeD($\mathcal{F}$,$\mathcal{V}$,$\mathcal{O}$) that makes $\poly(b, k ,1/\epsilon)$ calls to $G$, and runs in time $\poly(b,k,1/\epsilon, {\rm rt}_{G}(\poly(b,k,1/\epsilon)))$. If $OPT$ is the optimal obtainable expected value of $\obj$ for some BMeD instance, then the mechanism $M$ output by the algorithm on that instance yields {$\mathbb{E}[\obj(M)] \leq \frac{\alpha}{\beta} OPT + \epsilon$}, and is $\epsilon$-BIC. These guarantees hold with probability at least $1-\exp(\poly(b,k,1/\epsilon))$. \mattnote{Furthermore, the output mechanism is feasible, ex-post individually rational, and can be implemented in time $\poly(b, k, 1/\epsilon,{\rm rt}_G(\poly(b, k, 1/\epsilon)))$. These guarantees hold with probability $1$.}
\end{theorem}

\begin{remark}\label{rem:general}
Note that Theorem~\ref{thm:objective} extends to accommodate maximization problems as well. Here $\alpha \le 1 \le \beta$, and the guarantee  becomes $\mathbb{E}[\obj(M)] \geq \frac{\alpha}{\beta} OPT - \epsilon$. Theorem~\ref{thm:objective} also extends to settings where $\obj$ is a function of prices charged (such as revenue) and should be written as $\mathcal{O}(\vec{t},X,\vec{P})$, as well as settings where the prices charged affect whether or not a mechanism is feasible (such as simultaneously respecting budgets and individual rationality ex-post), as well as settings where $\obj$ behaves in a non-linear but convex (concave for maximization problems) way with respect to randomness (such as fractional fairness).
\end{remark}

There are two key techniques in proving Theorem~\ref{thm:objective}. The first is establishing an approximation preserving version of the celebrated ``equivalence of separation and optimization'' \cite{GrotschelLS81,KarpP80} that can accommodate $(\alpha,\beta)$-approximations, summarized in Section~\ref{sec:alphabeta} below. The second is understanding how to get mileage out of this improved equivalence within the framework of~\cite{CaiDW13b}, summarized in Section~\ref{sec:BMeD}.

\subsection{$(\alpha,\beta)$-approximate Equivalence of Separation and Optimization}\label{sec:alphabeta}
In this section we state and prove an extension of the ``equivalence of separation and optimization'' that accommodates $(\alpha,\beta)$-approximations. Let us first define what $(\alpha,\beta)$-approximations are for linear optimization problems.  In the definition below, $\alpha$ and $\beta$ are constants and $S$ is a subset of the coordinates. When we write the vector $(c\vec{x}_S,\vec{x}_{-S})$, we mean the vector $\vec{x}$ where all coordinates in $S$ have been multiplied by $c$. 

\begin{definition}
An algorithm $\mathcal{A}$ is an {\em $(\alpha,\beta,S)$-minimization algorithm} for polytope $P$ iff for any input vector $\vec{w}$ the vector $\mathcal{A}(\vec{w})$ output by the algorithm satisfies:
\begin{align}
(\beta \mathcal{A}(\vec{w})_S,\mathcal{A}(\vec{w})_{-S}) \cdot \vec{w} \leq \alpha \min_{\vec{x} \in P}\{\vec{x} \cdot \vec{w}\}.\label{eq:alpha beta general problem}
\end{align}
Given such algorithm, we also define the algorithm $\mathcal{A}_S^\beta$ that outputs $(\beta \mathcal{A}(\vec{w})_S,\mathcal{A}(\vec{w})_{-S})$.\footnote{We can similarly define the concept of a {\em $(\alpha,\beta,S)$-maximization algorithm} for polytope $P$ by flipping the inequality in~\eqref{eq:alpha beta general problem} and also switching $\min$ to $\max$.}
\end{definition}
\noindent Taking $\beta=1$ recovers the familiar notion of $\alpha$-approximation, except that we do not require the output of the algorithm to lie in $P$ in our definition. (Most meaningful applications of our framework will enforce this extra property though.) With $\beta <1$ (respectively $\beta>1$), the minimization (resp. maximization) becomes easier as the coordinates indexed by $S$ are discounted (boosted) by a factor of $\beta$ before comparing to $\alpha \cdot OPT$.

Our next theorem states that  an $(\alpha,\beta,S$)-approximation algorithm for polytope $P$ can be used to obtain a ``weird'' separation oracle (WSO) for $\alpha P$ (by $\alpha P$ we mean the polytope $P$ blown up by a factor of $\alpha$ or shrunk by a factor of $\alpha$, depending on whether $\alpha \geq 1$ or $\alpha \leq 1$), a concept explained in the statement of our theorem, which generalizes weird separation oracles from~\cite{CaiDW13}. 
We state and prove our theorem for $(\alpha,\beta,S)$-minimization algorithms, but the theorem holds for maximization by switching $\min$ to $\max$ and reversing the inequality in Property~\ref{property:converse}, and the proof is essentially identical. A complete proof of Theorem~\ref{thm:LPalphabeta} can be found in Appendix~\ref{app:proofs of alpha beta results}. Before stating the theorem formally, let's overview quickly what each property below is guaranteeing. Property 1 guarantees that the $WSO$ is consistent at least with respect to points inside $\alpha P$ (but may behave erratically outside of $\alpha P$). Property 2 guarantees that even though the points accepted by $WSO$ may not be in $\alpha P$ (or even in $P$), they will at least satisfy some relaxed notion of feasibility. Property 3 guarantees that $WSO$ terminates in polynomial time. Property 4 guarantees that if we run Ellipsoid with $WSO$ instead of a real separation oracle for $\alpha P$, that we don't sacrifice anything in terms of optimality \mattnote{(although the output is only guaranteed to satisfy the notion of feasibility given in Property 2. It may be infeasible in the traditional sense, i.e. not contained in $\alpha P$).}
\begin{theorem}\label{thm:LPalphabeta}
Let $P$ be a convex region in $\mathbb{R}^d$ and $\mathcal{A}$ an $(\alpha,\beta,S)$-minimization algorithm for $P$, for some $\alpha,\beta>0$. Then we can design a ``weird'' separation oracle $WSO$ for $\alpha P$ with the following properties:
\begin{enumerate}
\item Every halfspace output by $WSO$ will contain $\alpha P$.
\item Whenever $WSO(\vec{x}) =$ ``yes'' for some input $\vec{x}$, the execution of $WSO$ explicitly finds directions $\vec{w}_1,\ldots,\vec{w}_{d+1}$ such that $\vec{x} \in Conv\{\mathcal{A}^{\beta}_S(\vec{w}_1),\ldots,\mathcal{A}^{\beta}_S(\vec{w}_{d+1})\}$, and therefore $(\frac{1}{\beta}\vec{x}_S, \vec{x}_{-S}) \in Conv\{\mathcal{A}(\vec{w}_1),\ldots,\mathcal{A}(\vec{w}_{d+1})\}$ as well.
\item Let $b$ be the bit complexity of $\vec{x}$, $\ell$ an upper bound on the bit complexity of $\mathcal{A}(\vec{w})$ for all $\vec{w} \in [-1,1]^d$. Then on input $\vec{x}$, $WSO$ terminates in time $\poly(d,b,\ell,rt_{\mathcal{A}}(\poly(d,b,\ell)))$ and makes at most $\poly(d,b,\ell)$ queries to $\mathcal{A}$.

\item Let $Q$ be an arbitrary convex region in $\mathbb{R}^d$ described via some separation oracle, $\vec{c}$  a linear objective with $\vec{c}_{-S} = \vec{0}$, and $OPT = \min_{\vec{y} \in \alpha P \cap Q}\{\vec{c} \cdot \vec{y}\}$. Let also $\vec{z}$ be the output of the Ellipsoid algorithm for minimizing $\vec{c} \cdot \vec{y}$ over $\vec{y}\in\alpha P \cap Q$, but using $WSO$ as a separation oracle for $\alpha P$ instead of a standard separation oracle for $\alpha P$ {(i.e. use the exact same parameters for Ellipsoid as if $WSO$ was a valid separation oracle for $\alpha P$, and still use a standard separation oracle for $Q$)}. Then $\vec{c} \cdot \vec{z} \leq OPT$, and therefore $\vec{c} \cdot (\frac{1}{\beta}\vec{z}_S,\vec{z}_{-S}) \leq \frac{1}{\beta}OPT$. \label{property:converse}
\end{enumerate}
\end{theorem}

\subsection{Finding and Implementing Approximate Solutions to BMeD}\label{sec:BMeD}
In this section, we show how to make use of the techniques developed in Section~\ref{sec:alphabeta} to obtain approximate solutions to BMeD instances from $(\alpha,\beta)$-approximation algorithms for their corresponding GOOP instances. Essentially, the process breaks down into two parts: First, we solve a linear program to find an implicit form $\vec{\pi}_I$ whose $O$ component is at most $\alpha OPT + \epsilon$ (or at least $\alpha OPT - \epsilon$ for maximization). Second, we must actually implement $\vec{\pi}_I$ efficiently as a mechanism (recall that a mechanism takes as input a type profile $\vec{t}$ then selects a feasible outcome and charges prices). So ideally, $\vec{\pi}_I$ should be feasible (in all previous works~\cite{CaiDW12, CaiDW12b, CaiDW13, CaiDW13b}, this was the case). But with access only to an $(\alpha, \beta)$-approximation algorithm, obtaining an implicit form that is both feasible and within an $\alpha$-factor of $OPT$ simply isn't possible. We show instead that the guarantees of Theorem~\ref{thm:LPalphabeta} allow us to \emph{approximately} implement $\vec{\pi}_I$, losing only an additional factor of $\beta$ in the objective. We proceed to give more detail.

We start by writing a linear program (included as Figure~\ref{fig:LPBMeD} in Appendix~\ref{app:proofs of alpha beta results}) whose variables are components of an implicit form to find the truthful, feasible implicit form that optimizes $O$. A polynomial number of linear constraints enforce that the implicit form is truthful, but a separation oracle is required to enforce $\vec{\pi}_I \in F(\mathcal{F},\mathcal{D},\mathcal{O})$. Unfortunately, such a separation oracle can't be obtained efficiently, so we approximate $F(\mathcal{F},\mathcal{D},\mathcal{O})$ with another polytope for which we can obtain a separation oracle. Specifically, we define $\mathcal{D}'$ to be the uniform distribution over polynomially many samples (in $b, k$ and $1/\epsilon$) from $\mathcal{D}$, and approximate $F(\mathcal{F},\mathcal{D},\mathcal{O})$ with $F(\mathcal{F},\mathcal{D}',\mathcal{O})$.\footnote{Technically, the sampling procedure used to generate the support of $D'$ is slightly more involved, but we omit the details here because the distinction is not important. We refer the reader to~\cite{CaiDW12b} for more details.} The following proposition from~\cite{CaiDW13b} states this formally.

\begin{proposition}\label{prop:D'}(\cite{CaiDW12b},\cite{CaiDW13b})
With probability at least $1-\exp(\poly(b, k,1/\epsilon))$, for every feasible mechanism $M$, the implicit form of $M$ with respect to $\mathcal{D}$ ($\vec{\pi}^M \in F(\mathcal{F},\mathcal{D},\mathcal{O})$), and with respect to $\mathcal{D}'$ ($\vec{\pi}_0^M \in F(\mathcal{F},\mathcal{D}',\mathcal{O})$) are $\epsilon$-close. That is, $|\vec{\pi}^M - \vec{\pi}_0^M|_1 \leq \epsilon$. 
\end{proposition}

From here, we can \mattnote{replace $F(\mathcal{F},\mathcal{D},\mathcal{O})$ with $F(\mathcal{F},\mathcal{D}',\mathcal{O})$ to obtain the Linear Program in Figure~\ref{fig:LPD'} in Appendix~\ref{app:proofs of alpha beta results}, and} reduce the problem of solving this linear program and finding an approximately optimal implicit form to designing a weird separation oracle for $F(\mathcal{F},\mathcal{D}',\mathcal{O})$, using Theorem~\ref{thm:LPalphabeta}. Our progress so far is summarized in Proposition~\ref{prop:solveLP}. Note that, during the execution of the Ellipsoid algorithm referred to in the statement of the proposition, we will also be recording the auxiliary information produced by the weird separation oracle, as guaranteed by Property 2 of Theorem~\ref{thm:LPalphabeta}. 

\begin{proposition}\label{prop:solveLP}
With black-box access to an $(\alpha, \beta, \{O\})$-optimization algorithm, $\mathcal{A}$, for $F(\mathcal{F},\mathcal{D}',\mathcal{O})$, we can use the weird separation oracle of Theorem~\ref{thm:LPalphabeta} to run the Ellipsoid algorithm on the linear program of Figure~\ref{fig:LPD'}. The Ellipsoid algorithm will run in time $\poly(b,k,1/\epsilon, rt_\mathcal{A}(\poly(b, k,1/\epsilon)))$, and will output a truthful implicit form $\vec{\pi}_I$. With probability at least $1-\exp(\poly(b,k,1/\epsilon))$, the objective component $O$ of $\vec{\pi}_I$ will satisfy $O \leq \alpha OPT + \epsilon$, \mattnote{where $OPT$ is the expected value of $\mathcal{O}$ in the optimal truthful mechanism.} Furthermore, the algorithm will explicitly output a list of $d+1=\poly(b,k)$ directions $\vec{w}_1,\ldots,\vec{w}_{d+1}$ such that $\vec{\pi}_I \in  Conv\{\mathcal{A}^{\beta}_{\{O\}}(\vec{w}_1),\ldots,\mathcal{A}^{\beta}_{\{O\}}(\vec{w}_{d+1})\}$. 
\end{proposition}

Next, we show that one can obtain an $(\alpha, \beta, \{O\})$-optimization algorithm for $F(\mathcal{F},\mathcal{D}', \mathcal{O})$ given black-box access to an $(\alpha, \beta)$-approximation algorithm for $GOOP(\mathcal{F},\mathcal{V},\mathcal{O})$. The proof is included in Appendix~\ref{app:proofs of alpha beta results}.

\begin{proposition}\label{prop:goop}
Let all types in the support of $\mathcal{D}$ \mattnote{(and therefore $\mathcal{D}'$ as well)} be in the set $\mathcal{V}$. Then with black-box access to $G$, an $(\alpha, \beta)$-approximation algorithm for $GOOP(\mathcal{F},\mathcal{V},\mathcal{O})$, one can obtain an $(\alpha, \beta, \{O\})$-optimization algorithm, $\mathcal{A}$, for $F(\mathcal{F},\mathcal{D}',\mathcal{O})$. The algorithm runs in time polynomial in $b,k, 1/\epsilon$, and $rt_G(\poly(b,k,1/\epsilon))$. Furthermore, given as input any direction $\vec{w}$, one can implement in time $\poly(b,k,1/\epsilon, rt_G(\poly(b,k,1/\epsilon)))$ a feasible mechanism $M$ whose implicit form with respect to $\mathcal{D}'$, $\vec{\pi}_0^M$, satisfies $\vec{\pi}_0^M = \mathcal{A}(\vec{w})$.
\end{proposition}

Taken together, the above propositions provide an algorithm to find an implicit form $\vec{\pi}_I$ whose objective component is within an $\alpha$-factor of optimal. The only remaining step is to implement it. As mentioned earlier, there's a bit of a catch here because $\vec{\pi}_I$ may not even be feasible. We instead approximately implement $\vec{\pi}_I$, losing a factor of $\beta$ in the objective component but leaving the others untouched.

\begin{proposition}\label{prop:implement}
Let $\vec{\pi}_I = (O, \vec{\pi},\vec{p})$ be the implicit form and $\vec{w}_1,\ldots,\vec{w}_{d+1}$ the auxiliary information output by the algorithm of Proposition~\ref{prop:solveLP}, when using some $(\alpha, \beta)$-approximation algorithm $G$ for $GOOP(\mathcal{F},\mathcal{V},\mathcal{O})$ in order to get the required $(\alpha, \beta,\{O\})$-optimization algorithm  for $F(\mathcal{F},\mathcal{D}',\mathcal{O})$, via Proposition~\ref{prop:goop}. Then one can implement an ex-post individually rational mechanism $M$ in time polynomial in $d$ and $rt_G(\poly(b,k,1/\epsilon))$. With probability $1-\exp(\poly(b, k, 1/\epsilon))$, $M$ satisfies:
\begin{itemize}
\item $|O^M - O/\beta| \leq \epsilon$. 
\item $|\vec{\pi} - \vec{\pi}^M|_1 \leq \epsilon$.
\item $|\vec{p} - \vec{p}^M|_1 \leq \epsilon$. 
\end{itemize}
\end{proposition}

The proof of Proposition~\ref{prop:implement} is also included in Appendix~\ref{app:proofs of alpha beta results}. Combining Propositions~\ref{prop:solveLP} through~\ref{prop:implement} proves Theorem~\ref{thm:objective}. We formally show this in Appendix~\ref{app:proofs of alpha beta results} for completeness.

\section{Makespan and Fairness}\label{sec:makespan}
In this section we state bicriterion algorithmic results for minimizing makespan and maximizing fairness, as well as their implications to mechanism design. Additionally, we state a general theorem that is useful in developing $(\alpha, \beta)$-approximation algorithms via algorithms that round fractional solutions. In particular, this theorem allows us to get a $(\frac{1}{2},\tilde{O}(\sqrt{k}))$-approximation for fairness based on the algorithm of Asadpour and Saberi~\cite{AsadpourS07}. We begin by stating our $(\alpha, \beta)$-guarantees.

\begin{theorem}\label{thm:alphabetamakespan}
There is a polynomial-time $(1,\frac{1}{2})$-approximation for minimizing makespan with costs on unrelated machines. The algorithm is based on a rounding theorem of Shmoys and Tardos~\cite{ShmoysT93}.
\end{theorem}

\begin{theorem}\label{thm:alphabetafairness}
There is a polynomial-time algorithm for maximizing fairness with costs on unrelated machines with the following guarantees:
\begin{itemize}
\item A $(\frac{1}{2},\tilde{O}(\sqrt{k}))$-approximation based on an algorithm of Asadpour and Saberi~\cite{AsadpourS07}.
\item A $(1, m-k+1)$-approximation based on an algorithm of Bezakova and Dani~\cite{BezakovaD05}.
\end{itemize}
\end{theorem}

Next, we state a theorem useful in the analysis of algorithms that round fractional solutions. In the theorem statement below, $\vec{x} \in [0,1]^{km}$ denotes a fractional assignment of jobs to machines, and $\vec{y}$ denotes a randomly sampled assignment in $\{0,1\}^{km}$ according to some rounding procedure. Note that when we write $F(\vec{y})$, we mean the expected value of the random variable $F(\vec{y})$, and \emph{not} the fairness computed with respect to the fractional assignment $\mathbb{E}[y_{ij}]$. Finally, $\vec{v} \in \{0,1\}^{km}$ denotes the integral allocation that maximizes $C(\vec{w})$ over all $\vec{w} \in \{0,1\}^{km}$. In other words, $\vec{v}$ assigns each job to the machine with the highest non-negative cost, if one exists (and nowhere if none exists). We emphasize again that Theorem~\ref{thm:useful} is what enables the first bullet in Theorem~\ref{thm:alphabetafairness} above, and will likely have applications to the analysis of other potential $(\alpha,\beta)$-approximation algorithms. We also note that Theorem~\ref{thm:useful} applies to any maximization objective, not just fairness.

\begin{theorem}\label{thm:useful}
Let $\vec{x} \in [0,1]^{km}$ be a fractional assignment of jobs to machines that is an $(\alpha, \beta)$-approximation with respect to the optimal integral assignment (that is, $\beta F(\vec{x}) + C(\vec{x}) \geq \alpha OPT$), and $\vec{y}$ a random variable of assignments supported on $\{0,1\}^{km}$ satisfying $z\cdot F(\vec{y}) \geq  F(\vec{x})$, for some $z\ge1$, and $\mathbb{E}[y_{ij}] \leq x_{ij}$ for all $i,j$. Then for any $\gamma \in [0,1]$, at least one of the following is true:
\begin{itemize}
\item $z \cdot \beta F(\vec{y}) + C(\vec{y}) \geq \gamma (\beta F(\vec{x}) + C(\vec{x})) \geq \gamma\cdot \alpha OPT$. That is, $\vec{y}$ is a $(\gamma \alpha, z \beta)$-approximation.
\item Or $F(\vec{v}) + C(\vec{v}) \geq (1-\gamma)(\beta F(\vec{x}) + C(\vec{x})) \geq (1-\gamma) \cdot \alpha OPT$. That is, $\vec{v}$ is a $((1-\gamma)\alpha, 1)$-approximation. $\vec{v}$ is the assignment maximizing $C(\vec{w})$ over all feasible $\vec{w} \in \{0,1\}^{km}$.
\end{itemize}
\end{theorem}

Proofs of Theorems~\ref{thm:alphabetamakespan},~\ref{thm:alphabetafairness}, and~\ref{thm:useful} can be found in Appendix~\ref{app:makespanfairness}. We conclude this section by stating formally our results on truthful mechanisms for scheduling on unrelated machines. Theorems~\ref{thm:truthfulmakespan} and~\ref{thm:truthfulfairness} are direct corollaries of Theorem~\ref{thm:objective} combined with Theorems~\ref{thm:alphabetamakespan} and~\ref{thm:alphabetafairness}.

\begin{theorem}\label{thm:truthfulmakespan}
There is a polynomial-time 2-approximation algorithm for truthfully minimizing makespan on unrelated machines (formally the problem {\rm BMeD(}$\{0,1\}^{km}$, \{additive functions with coefficients $p_{ij} \in [0,1]$\}, Makespan{\rm )}. For any desired $\epsilon>0$, \mattnote{the output mechanism  is $\epsilon$-BIC, and has expected makespan at most $2OPT + \epsilon$ with probability at least $1-\exp(\poly(b,k,m,1/\epsilon))$. Furthermore, the output mechanism is feasible, can be implemented in time $\poly(b,k,m,1/\epsilon)$, and is ex-post individually rational with probability $1$. The runtime of the algorithm is $\poly(b,k,m,1/\epsilon)$. }
\end{theorem}

\begin{theorem}\label{thm:truthfulfairness}
There is a polynomial-time $\min\{\tilde{O}(\sqrt{k}),m-k+1\}$-approximation algorithm for truthfully maximizing fairness on unrelated machines (formally the problem {\rm BMeD(}$\{0,1\}^{km}$, \{additive functions with coefficients $p_{ij} \in [0,1]$\}, Fairness{\rm )}. For any desired $\epsilon>0$, \mattnote{the output mechanism  is $\epsilon$-BIC, and has expected fairness at least $\min\{\tilde{O}(\sqrt{k}),m-k+1\}OPT + \epsilon$ with probability at least $1-\exp(\poly(b,k,m,1/\epsilon))$. Furthermore, the output mechanism is feasible, can be implemented in time $\poly(b,k,m,1/\epsilon)$, and is ex-post individually rational with probability $1$. The runtime of the algorithm is $\poly(b,k,m,1/\epsilon)$. }
\end{theorem}

\bibliographystyle{abbrv}
\bibliography{costasbibold} 

\begin{thebibliography}{10}

\bibitem{ArcherT01}
A.~Archer and {\'E}.~Tardos.
\newblock {Truthful Mechanisms for One-Parameter Agents}.
\newblock In {\em the 42nd Annual Symposium on Foundations of Computer Science
  (FOCS)}, 2001.

\bibitem{AsadpourFS08}
A.~Asadpour, U.~Feige, and A.~Saberi.
\newblock Santa claus meets hypergraph matchings.
\newblock In {\em the 12th International Workshop on Approximation,
  Randomization, and Combinatorial Optimization (APPROX-RANDOM)}, 2008.

\bibitem{AsadpourS07}
A.~Asadpour and A.~Saberi.
\newblock An approximation algorithm for max-min fair allocation of indivisible
  goods.
\newblock In {\em the 39th Annual ACM Symposium on Theory of Computing (STOC)},
  2007.

\bibitem{AshlagiDL12}
I.~Ashlagi, S.~Dobzinski, and R.~Lavi.
\newblock Optimal lower bounds for anonymous scheduling mechanisms.
\newblock {\em Mathematics of Operations Research}, 37(2):244--258, 2012.

\bibitem{BansalS06}
N.~Bansal and M.~Sviridenko.
\newblock The santa claus problem.
\newblock In {\em the 38th annual ACM symposium on Theory of computing (STOC)},
  pages 31--40, 2006.

\bibitem{BateniCG09}
M.~Bateni, M.~Charikar, and V.~Guruswami.
\newblock Maxmin allocation via degree lower-bounded arborescences.
\newblock In {\em the 41st annual ACM symposium on Theory of Computing (STOC)},
  pages 543--552, 2009.

\bibitem{BateniHSZ13}
M.~Bateni, N.~Haghpanah, B.~Sivan, and M.~Zadimoghaddam.
\newblock Revenue maximization with nonexcludable goods.
\newblock In {\em the 9th International Conference on Web and Internet
  Economics (WINE)}, 2013.

\bibitem{BeiH11}
X.~Bei and Z.~Huang.
\newblock {Bayesian Incentive Compatibility via Fractional Assignments}.
\newblock In {\em the Twenty-Second Annual ACM-SIAM Symposium on Discrete
  Algorithms (SODA)}, 2011.

\bibitem{BezakovaD05}
I.~Bez{\'a}kov{\'a} and V.~Dani.
\newblock Allocating indivisible goods.
\newblock {\em SIGecom Exchanges}, 5(3):11--18, 2005.

\bibitem{BramsT96}
S.~J. Brams and A.~D. Taylor.
\newblock {\em Fair Division: From cake-cutting to dispute resolution}.
\newblock Cambridge University Press, 1996.

\bibitem{BuchfuhrerDFKMPSSU10}
D.~Buchfuhrer, S.~Dughmi, H.~Fu, R.~Kleinberg, E.~Mossel, C.~H. Papadimitriou,
  M.~Schapira, Y.~Singer, and C.~Umans.
\newblock {Inapproximability for VCG-Based Combinatorial Auctions}.
\newblock In {\em Proceedings of the Twenty-First Annual ACM-SIAM Symposium on
  Discrete Algorithms (SODA)}, 2010.

\bibitem{CaiDW12}
Y.~Cai, C.~Daskalakis, and S.~M. Weinberg.
\newblock {An Algorithmic Characterization of Multi-Dimensional Mechanisms}.
\newblock In {\em the 44th Annual ACM Symposium on Theory of Computing (STOC)},
  2012.

\bibitem{CaiDW12b}
Y.~Cai, C.~Daskalakis, and S.~M. Weinberg.
\newblock {Optimal Multi-Dimensional Mechanism Design: Reducing Revenue to
  Welfare Maximization}.
\newblock In {\em the 53rd Annual IEEE Symposium on Foundations of Computer
  Science (FOCS)}, 2012.

\bibitem{CaiDW13}
Y.~Cai, C.~Daskalakis, and S.~M. Weinberg.
\newblock {Reducing Revenue to Welfare Maximization: Approximation Algorithms
  and other Generalizations}.
\newblock In {\em the 24th Annual ACM-SIAM Symposium on Discrete Algorithms
  (SODA)}, 2013.
\newblock http://arxiv.org/pdf/1305.4000v1.pdf.

\bibitem{CaiDW13b}
Y.~Cai, C.~Daskalakis, and S.~M. Weinberg.
\newblock {Understanding Incentives: Mechanism Design becomes Algorithm
  Design}.
\newblock In {\em the 54th Annual IEEE Symposium on Foundations of Computer
  Science (FOCS)}, 2013.
\newblock http://arxiv.org/pdf/1305.4002v1.pdf.

\bibitem{ChakrabartyCK09}
D.~Chakrabarty, J.~Chuzhoy, and S.~Khanna.
\newblock On allocating goods to maximize fairness.
\newblock In {\em 50th Annual IEEE Symposium on Foundations of Computer Science
  (FOCS)}, 2009.

\bibitem{ChawlaHMS13}
S.~Chawla, J.~Hartline, D.~Malec, and B.~Sivan.
\newblock {Prior-Independent Mechanisms for Scheduling}.
\newblock In {\em Proceedings of 45th ACM Symposium on Theory of Computing
  (STOC)}, 2013.

\bibitem{ChawlaIL12}
S.~Chawla, N.~Immorlica, and B.~Lucier.
\newblock On the limits of black-box reductions in mechanism design.
\newblock In {\em Proceedings of the 44th Symposium on Theory of Computing
  (STOC)}, 2012.

\bibitem{ChristodoulouKK07}
G.~Christodoulou, E.~Koutsoupias, and A.~Kov{\'a}cs.
\newblock {Mechanism Design for Fractional Scheduling on Unrelated Machines}.
\newblock In {\em the 34th International Colloquium on Automata, Languages and
  Programming (ICALP)}, 2007.

\bibitem{ChristodoulouKV07}
G.~Christodoulou, E.~Koutsoupias, and A.~Vidali.
\newblock {A Lower Bound for Scheduling Mechanisms}.
\newblock In {\em the 18th Annual ACM-SIAM Symposium on Discrete Algorithms
  (SODA)}, 2007.

\bibitem{ChristodoulouK10}
G.~Christodoulou and A.~Kov{\'a}cs.
\newblock {A Deterministic Truthful PTAS for Scheduling Related Machines}.
\newblock In {\em the 21st Annual ACM-SIAM Symposium on Discrete Algorithms
  (SODA)}, 2010.

\bibitem{DaskalakisW12}
C.~Daskalakis and S.~M. Weinberg.
\newblock {Symmetries and Optimal Multi-Dimensional Mechanism Design}.
\newblock In {\em the 13th ACM Conference on Electronic Commerce (EC)}, 2012.

\bibitem{DhangwatnotaiDDR08}
P.~Dhangwatnotai, S.~Dobzinski, S.~Dughmi, and T.~Roughgarden.
\newblock {Truthful Approximation Schemes for Single-Parameter Agents}.
\newblock In {\em the 49th Annual IEEE Symposium on Foundations of Computer
  Science (FOCS)}, 2008.

\bibitem{Dobzinski11}
S.~Dobzinski.
\newblock {An Impossibility Result for Truthful Combinatorial Auctions with
  Submodular Valuations}.
\newblock In {\em Proceedings of the 43rd ACM Symposium on Theory of Computing
  (STOC)}, 2011.

\bibitem{DobzinskiV12}
S.~Dobzinski and J.~Vondrak.
\newblock {The Computational Complexity of Truthfulness in Combinatorial
  Auctions}.
\newblock In {\em Proceedings of the ACM Conference on Electronic Commerce
  (EC)}, 2012.

\bibitem{GrotschelLS81}
M.~Gr{\"o}tschel, L.~Lov{\'a}sz, and A.~Schrijver.
\newblock {The Ellipsoid Method and its Consequences in Combinatorial
  Optimization}.
\newblock {\em Combinatorica}, 1(2):169--197, 1981.

\bibitem{HaghpanahIMM11}
N.~Haghpanah, N.~Immorlica, V.~S. Mirrokni, and K.~Munagala.
\newblock Optimal auctions with positive network externalities.
\newblock In {\em the 12th ACM Conference on Electronic Commerce (EC)}, 2011.

\bibitem{HartlineKM11}
J.~D. Hartline, R.~Kleinberg, and A.~Malekian.
\newblock {Bayesian Incentive Compatibility via Matchings}.
\newblock In {\em the Twenty-Second Annual ACM-SIAM Symposium on Discrete
  Algorithms (SODA)}, 2011.

\bibitem{HartlineL10}
J.~D. Hartline and B.~Lucier.
\newblock {Bayesian Algorithmic Mechanism Design}.
\newblock In {\em the 42nd ACM Symposium on Theory of Computing (STOC)}, 2010.

\bibitem{Hochbaum96}
D.~S. Hochbaum.
\newblock {\em {Approximation algorithms for NP-hard problems}}.
\newblock PWS Publishing Co., 1996.

\bibitem{KarpP80}
R.~M. Karp and C.~H. Papadimitriou.
\newblock {On Linear Characterizations of Combinatorial Optimization Problems}.
\newblock In {\em the 21st Annual Symposium on Foundations of Computer Science
  (FOCS)}, 1980.

\bibitem{KhotP07}
S.~Khot and A.~K. Ponnuswami.
\newblock Approximation algorithms for the max-min allocation problem.
\newblock In {\em the 11th International Workshop on Approximation,
  Randomization, and Combinatorial Optimization (APPROX-RANDOM)}, 2007.

\bibitem{Knaster46}
B.~Knaster.
\newblock Sur le probleme du partage pragmatique de h. steinhaus.
\newblock In {\em Annales de la Societ{\'e} Polonaise de Mathematique},
  volume~19, pages 228--230, 1946.

\bibitem{KoutsoupiasV07}
E.~Koutsoupias and A.~Vidali.
\newblock {A Lower Bound of $1+\phi$ for Truthful Scheduling Mechanisms}.
\newblock In {\em the 32nd International Symposium on the Mathematical
  Foundations of Computer Science (MFCS)}, 2007.

\bibitem{LenstraST87}
J.~K. Lenstra, D.~B. Shmoys, and {\'E}.~Tardos.
\newblock Approximation algorithms for scheduling unrelated parallel machines.
\newblock In {\em FOCS}, 1987.

\bibitem{Lu09}
P.~Lu.
\newblock {On 2-Player Randomized Mechanisms for Scheduling}.
\newblock In {\em the 5th International Workshop on Internet and Network
  Economics (WINE)}, 2009.

\bibitem{LuY08a}
P.~Lu and C.~Yu.
\newblock {An Improved Randomized Truthful Mechanism for Scheduling Unrelated
  Machines}.
\newblock In {\em the 25th Annual Symposium on Theoretical Aspects of Computer
  Science (STACS)}, 2008.

\bibitem{LuY08b}
P.~Lu and C.~Yu.
\newblock {Randomized Truthful Mechanisms for Scheduling Unrelated Machines}.
\newblock In {\em the 4th International Workshop on Internet and Network
  Economics (WINE)}, 2008.

\bibitem{MualemS07}
A.~Mu'alem and M.~Schapira.
\newblock Setting lower bounds on truthfulness: extended abstract.
\newblock In {\em the 18th Annual ACM-SIAM Symposium on Discrete Algorithms
  (SODA)}, 2007.

\bibitem{Myerson81}
R.~B. Myerson.
\newblock {Optimal Auction Design}.
\newblock {\em Mathematics of Operations Research}, 6(1):58--73, 1981.

\bibitem{NisanR99}
N.~Nisan and A.~Ronen.
\newblock {Algorithmic Mechanism Design (Extended Abstract)}.
\newblock In {\em Proceedings of the Thirty-First Annual ACM Symposium on
  Theory of Computing (STOC)}, 1999.

\bibitem{PapadimitriouSS08}
C.~H. Papadimitriou, M.~Schapira, and Y.~Singer.
\newblock On the hardness of being truthful.
\newblock In {\em Proceedings of the 49th Annual IEEE Symposium on Foundations
  of Computer Science (FOCS)}, 2008.

\bibitem{ShmoysT93b}
D.~B. Shmoys and {\'E}.~Tardos.
\newblock An approximation algorithm for the generalized assignment problem.
\newblock {\em Mathematical Programming}, 62(1-3):461--474, 1993.

\bibitem{ShmoysT93}
D.~B. Shmoys and {\'E}.~Tardos.
\newblock {Scheduling Unrelated Machines with Costs}.
\newblock In {\em the 4th Symposium on Discrete Algorithms (SODA)}, 1993.

\bibitem{Steinhaus48}
H.~Steinhaus.
\newblock The problem of fair division.
\newblock {\em Econometrica}, 16(1), 1948.

\end{thebibliography}
\appendix
\section{Proof of Theorems~\ref{thm:objective} and~\ref{thm:LPalphabeta}} \label{app:proofs of alpha beta results}
\subsection{Proof of Theorem~\ref{thm:LPalphabeta}}\label{sec:min}
Here we prove Theorem~\ref{thm:LPalphabeta} for minimization algorithms, noting that the proof for maximization algorithms is nearly identical after switching $\leq$ for $\geq$ and $\min$ for $\max$ where appropriate. Much of the proof is similar to that of Theorem~H.1 in~\cite{CaiDW13}. We include it here for completeness, however we refer the reader to~\cite{CaiDW13} for the proof of some technical lemmas. We begin by defining the weird separation oracle in Figure~\ref{fig:WSO}. This is identical to the weird separation oracle used in~\cite{CaiDW13}, except that we use $\mathcal{A}^\beta_S$ instead of just $\mathcal{A}$. 

		\begin{figure}[h!]
	\colorbox{MyGray}{
	\begin{minipage}{\textwidth} {	
		$WSO(\vec{y})=$
	\begin{itemize}	
	\item ``\textbf{Yes}'' if the ellipsoid algorithm with $N$ iterations\footnote{The appropriate choice of $N$ for our use of $WSO$ is provided in Corollary~5.1 of Section~5 in~\cite{CaiDW13b}. $N$ is polynomial in the appropriate quanitites.} outputs ``infeasible'' on the following problem:
			 
			  \underline{\textbf{variables:}} $\vec{w}, t$;\\
			 \underline{\textbf{constraints:}}
				\begin{itemize}
		 		\item $\vec{w} \in [-1,1]^d;$
		 		\item $t \geq \vec{y} \cdot \vec{w}+ \delta$;\footnote{The appropriate choice of $\delta$ for our use of $WSO$ is provided in Lemma~5.1 of Section~5 in~\cite{CaiDW13}. The bit complexity of $\delta$ is polynomial in the appropriate quantities.}
		 		\item $\widehat{WSO}(\vec{w},t) = $
		 		\begin{itemize}
		 		\item ``yes'' if $t \leq \mathcal{A}^\beta_S(\vec{w}) \cdot \vec{w}$;\footnote{Notice that the set $\{(\vec{w},t)|\widehat{WSO}(\vec{w},t) =$ ``Yes''$\}$ is not necessarily convex or even connected.}
		 		\item the violated hyperplane $t' \leq \mathcal{A}^\beta_S(\vec{w})\cdot \vec{w}'$ otherwise.
		 		\end{itemize}
				 		\end{itemize}
	\item If a feasible point $(t^*,\vec{w}^*)$ is found, output the violated hyperplane $\vec{w}^* \cdot \vec{x} \leq t^*$.
	\end{itemize}		}
			\end{minipage}} \caption{A ``weird'' separation oracle.}\label{fig:WSO}
	\end{figure}

\begin{lemma}\label{lem:one}
If $\mathcal{A}^\beta_S$ is an $(\alpha, \beta, S)$-minimization algorithm for the polytope $P$, then every halfspace output by $WSO$ contains $\alpha P$.
\end{lemma}
\begin{proof}
If $WSO$ outputs a halfspace $(\vec{w}^*, t^*)$, then we must have $\widehat{WSO}(\vec{w}^*,t^*) = \text{``yes''}$, implying that $t^* \leq \mathcal{A}^\beta_S(\vec{w}^*) \cdot \vec{w}^*$. Because $\mathcal{A}$ is an $(\alpha, \beta, S)$-approximation, we know that $\mathcal{A}^\beta_S(\vec{w}^*) \cdot \vec{w}^* \leq \alpha \min_{\vec{x} \in P} \{\vec{x} \cdot \vec{w}^*\} = \min_{\vec{x} \in \alpha P} \{\vec{x} \cdot \vec{w}^*\}$. Therefore, every $\vec{x} \in \alpha P$ satisfies $t^* \leq \vec{w}^* \cdot \vec{x}$, and the halfspace contains $\alpha P$. 
\end{proof}

\begin{lemma}\label{lem:three}
Let $b$ be the bit complexity of $\vec{y}$, $\ell$ an upper bound on the bit complexity of $\mathcal{A}^\beta_S(\vec{w})$ for all $\vec{w} \in [-1,1]^d$. Then on input $\vec{y}$, $WSO$ terminates in time $\poly(d, b, \ell, rt_\mathcal{A}(\poly(d,b,\ell)))$ and makes at most $\poly(d,b,\ell)$ queries to $\mathcal{A}$.
\end{lemma}

\begin{proof}
It is clear that $\widehat{WSO}$ is queried at most $\poly(N, d, b, \ell, BC(\delta))$ times, where $BC(\delta)$ is the bit complexity of $\delta$. Each execution of $\widehat{WSO}$ makes one call to $\mathcal{A}$. So as long as $N$ and the bit complexity of $\delta$ are both polynomial in $\poly(d, b, \ell)$, the lemma holds. This is shown in Corollary~5.1 of Section~5 and Lemma~5.1 of Section~5 in~\cite{CaiDW13}, and we omit further details here. 
\end{proof}

\begin{lemma}\label{lem:four}
Let $Q$ be an arbitrary convex region in $\mathbb{R}^d$ described via some separation oracle, and $OPT = \min_{\vec{y} \in \alpha P \cap Q}\{\vec{c} \cdot \vec{y}\}$ for some vector $\vec{c}$. Let also $\vec{z}$ be the output of the Ellipsoid algorithm for minimizing $\vec{c} \cdot \vec{y}$ over $\vec{y} \in \alpha P \cap Q$, but using $WSO$ as a separation oracle for $\alpha P$ instead of a standard separation oracle (i.e. use the exact parameters for Ellipsoid as if $WSO$ was a valid separation oracle for $\alpha P$, and still use a standard separation oracle for $Q$). Then $\vec{z} \cdot \vec{c} \leq OPT$. 
\end{lemma}

\begin{proof}
When the Ellipsoid algorithm tries to minimize $\vec{c} \cdot \vec{x}$, it does a binary search over possible values $C$, and checks whether or not there is a point $\vec{x}$ satisfying $\vec{x} \cdot \vec{c} \leq C$, $\vec{x} \in Q$, and $WSO(\vec{x}) = \text{``yes''}$. If there is a point $\vec{x}$ satisfying $\vec{x} \cdot \vec{c} \leq C$, $\vec{x} \in Q$, and $\vec{x} \in \alpha P$, then clearly every halfspace output by the separation oracle for $Q$ contains $\vec{x}$, and so does the halfspace $\vec{y} \cdot \vec{c} \leq C$. Furthermore, by Lemma~\ref{lem:one}, every halfspace output by $WSO$ contains $\vec{x}$ as well. Therefore, if $OPT \leq C$, the Ellipsoid algorithm using $WSO$ will find a feasible point and continue its binary search. Therefore, the algorithm must conclude with a point $\vec{z}$ satisfying $\vec{z} \cdot \vec{c} \leq OPT$.
\end{proof}

\begin{lemma}\label{lem:two}
Whenever $WSO(\vec{y}) =$ ``yes'' for some input $\vec{y}$, the execution of $WSO$ explicitly finds directions $\vec{w}_1,\ldots,\vec{w}_{d+1}$ such that $\vec{y} \in Conv\{\mathcal{A}^\beta_S(\vec{w}_1),\ldots,\mathcal{A}^\beta_S(\vec{w}_{d+1})\}$. 
\end{lemma}

\begin{proof}
Consider the following intersection of halfspaces:
$$t' \geq \vec{y} \cdot \vec{w}' +\delta$$
$$ t' \leq \mathcal{A}^\beta_S(\vec{w}) \cdot \vec{w}',\ \forall \vec{w} \in W$$
$$\vec{w}' \in [-1,1]^d$$

If $\vec{y} \notin Conv(\{\mathcal{A}^\beta_S(\vec{w}) | \vec{w} \in W\})$, there exists some weight vector $\vec{w}^* \in [-1,1]^d$ such that $\vec{y} \cdot \vec{w}^* < \min_{\vec{w} \in W}\{\mathcal{A}^\beta_S(\vec{w})\cdot \vec{w}^*\}$. And for appropriately chosen $\delta = 2^{-\poly(d, b, \ell)}$, we also have $\vec{y} \cdot \vec{w}^* \leq \min_{\vec{w} \in W} \{\mathcal{A}^\beta_S(\vec{w})\cdot \vec{w}^*\} - \delta$. 

So if $\vec{y} \notin Conv(\{\mathcal{A}^\beta_S(\vec{w}) | \vec{w} \in W\})$, consider the point $\vec{w}^*, t^*$, with $t^* = \min_{\vec{w} \in W} \{\mathcal{A}^\beta_S(\vec{w}) \cdot \vec{w}^*\}$. By the reasoning in the previous paragraph, it's clear that $(\vec{w}^*, t^*)$ is in the above intersection of halfspaces. 

So consider an execution of $WSO$ that accepts the point $\vec{y}$. Then taking $W$ to be the set of directions $\vec{w}$ queried by the Ellipsoid algorithm during the execution of $WSO$, the Ellipsoid algorithm deemed the intersection of halfspaces above to be empty. This necessarily means that $\vec{y} \in Conv(\{\mathcal{A}^\beta_S(\vec{w})| \vec{w} \in W\})$, as otherwise the previous paragraphs prove that the above intersection of halfspaces wouldn't be empty.

So we may take $\vec{w}_1,\ldots,\vec{w}_{d+1}$ to be (an appropriately chosen subset of) the directions queried by the Ellipsoid algorithm during the execution of $WSO$ and complete the proof of the lemma.
\end{proof}

It is clear that each lemma proves one guarantee of Theorem~\ref{thm:LPalphabeta}, completing the proof.

\subsection{Theorem~\ref{thm:objective}}
We begin by stating the linear program used in our algorithm in Figure~\ref{fig:LPBMeD}, \mattnote{and it's modification to use $F(\mathcal{F}, \mathcal{D}',\mathcal{O})$ instead in Figure~\ref{fig:LPD'}} (both taken directly from~\cite{CaiDW13b}).
\begin{figure}[ht]
\colorbox{MyGray}{
\begin{minipage}{\textwidth} {
\noindent\textbf{Variables:}
\begin{itemize}
\item $\pi_i(t,t')$, for all bidders $i$ and types $t,t' \in T_i$, denoting the expected value obtained by bidder $i$ when their true type is $t$ but they report $t'$ instead.
\item $p_i(t)$, for all bidders $i$ and types $t \in T_i$, denoting the expected price paid by bidder $i$ when they report type $t$.
\item $O$, denoting the expected value of $\obj$.
\end{itemize}
\textbf{Constraints:}
\begin{itemize}
\item $\pi_i(t,t) - p_i(t) \geq \pi_i(t,t') - p_i(t')$, for all bidders $i$, and types $t,t' \in T_i$, guaranteeing that the implicit form {$(O,\vec{\pi},\vec{p})$} is BIC.
\item $\pi_i(t,t) - p_i(t) \geq 0$, for all bidders $i$, and types $t \in T_i$, guaranteeing that the implicit form {$(O,\vec{\pi},\vec{p})$} is individually rational.
\item $(O,\vec{\pi})\in F(\mathcal{F},\mathcal{D},\obj)$, guaranteeing that the implicit form $(O,\vec{\pi},\vec{p})$ is feasible.
\end{itemize}
\textbf{Minimizing:}
\begin{itemize}
\item $O$, the expected value of $\obj$ when played truthfully by bidders sampled from $\mathcal{D}$.\\
\end{itemize}}
\end{minipage}}
\caption{A linear programming formulation for BMeD.}
\label{fig:LPBMeD}
\end{figure}

\begin{observation}(\cite{CaiDW13b})
Any solution to the linear program of Figure~\ref{fig:LPBMeD} is the implicit form of a feasible, BIC, IR mechanism that achieves the optimal expected value of $\mathcal{O}$.
\end{observation}

\begin{figure}[ht]
\colorbox{MyGray}{
\begin{minipage}{\textwidth} {
\noindent\textbf{Variables:}
\begin{itemize}
\item $\pi_i(t,t')$, for all bidders $i$ and types $t,t' \in T_i$, denoting the expected value obtained by bidder $i$ when their true type is $t$ but they report $t'$ instead.
\item $p_i(t)$, for all bidders $i$ and types $t \in T_i$, denoting the expected price paid by bidder $i$ when they report type $t$.
\item $O$, denoting the expected value of $\obj$.
\end{itemize}
\textbf{Constraints:}
\begin{itemize}
\item $\pi_i(t,t) - p_i(t) \geq \pi_i(t,t') - p_i(t')$, for all bidders $i$, and types $t,t' \in T_i$, guaranteeing that the implicit form {$(O,\vec{\pi},\vec{p})$} is BIC.
\item $\pi_i(t,t) - p_i(t) \geq 0$, for all bidders $i$, and types $t \in T_i$, guaranteeing that the implicit form {$(O,\vec{\pi},\vec{p})$} is individually rational.
\item $(O,\vec{\pi})\in F(\mathcal{F},\mathcal{D}',\obj)$, guaranteeing that the implicit form $(O,\vec{\pi},\vec{p})$ is (almost) feasible.
\end{itemize}
\textbf{Minimizing:}
\begin{itemize}
\item $O$, (almost) the expected value of $\obj$ when played truthfully by bidders sampled from $\mathcal{D}$.\\
\end{itemize}}
\end{minipage}}
\caption{An $\epsilon$-approximate linear programming formulation for BMeD.}
\label{fig:LPD'}
\end{figure}

We now proceed by proving Propositions~\ref{prop:solveLP} through~\ref{prop:implement}.

\begin{prevproof}{Proposition}{prop:solveLP}
Theorem~\ref{thm:LPalphabeta} guarantees that the linear program can be solved in the desired runtime, and that the desired directions $\vec{w}_1,\ldots,\vec{w}_{d+1}$ will be output. It is clear that any implicit form satisfying the constraints is truthful. 

Let now $OPT'$ denote the value of the LP in Figure~\ref{fig:LPD'}, $OPT_\alpha$ denote the value of the LP in Figure~\ref{fig:LPBMeD} using a real separation oracle for $\alpha F(\mathcal{F},\mathcal{D},\mathcal{O})$, and $OPT'_\alpha$ denote the value of the LP in Figure~\ref{fig:LPD'} using a real separation oracle for $\alpha F(\mathcal{F}, \mathcal{D}', \mathcal{O})$. Theorem~\ref{thm:LPalphabeta} also guarantees that $O \leq OPT'_\alpha$. So we just need to show that $OPT'_\alpha \leq \alpha OPT + \epsilon$ with the desired probability.

To see this, first observe that the origin satisfies every constraint in the linear program not due to $F(\mathcal{F},\mathcal{D}',\mathcal{O})$ (i.e. the truthfulness constraints) with equality. Therefore, if any implicit form $\vec{\pi}_I$ is truthful, so is the implicit form $\alpha \vec{\pi}_I$. This immediately implies that $OPT_\alpha = \alpha OPT$. 

By Proposition~\ref{prop:D'}, we know that with the desired probability, the implicit form (with respect to $\mathcal{D}$) of whatever mechanism implements $\vec{\pi}_I$ (with respect to $\mathcal{D}'$) is $\epsilon$-close to $\vec{\pi}_I$, and therefore $OPT'_\alpha \leq OPT_\alpha + \epsilon = \alpha OPT + \epsilon$ with the desired probability as well.
\end{prevproof}

In order to prove Proposition~\ref{prop:goop}, we make use of a technical lemma from~\cite{CaiDW13b} (specifically, combining Propositions~1 and~7).

\begin{proposition}\label{prop:VW}(\cite{CaiDW13b})
Let $\vec{w}$ be a direction in $[-1,1]^{1+\sum_i (|T_i|^2 + |T_i|)}$. Define the virtual objective $\mathcal{O}'_{\vec{w}}$ as:

$$\mathcal{O}'_{\vec{w}}(\vec{t}', X) = w_O \cdot \mathcal{O}(\vec{t}',X) + \sum_i \sum_{t \in T_i} \frac{w_i(t, t')}{\Pr[t']} \cdot t(X)$$

Then for any mechanism $M = (A, P)$, if $\vec{\pi}_I^M$ denotes the implicit form of $M$ with respect to $\mathcal{D}$, $\vec{\pi}_I^M \cdot \vec{w}$ is exactly the expected virtual objective of $M$ on type profiles sampled from $\mathcal{D}$. Formally:

$$\vec{\pi}^M_I \cdot \vec{w} = \mathbb{E}_{\vec{t}' \leftarrow \mathcal{D}} [\mathcal{O}'_{\vec{w}}(\vec{t}', A(\vec{t}'))]$$
\end{proposition}

\begin{prevproof}{Proposition}{prop:goop}
Let's first consider the case that $w_O \geq 0$. The $w_O < 0$ case will be handled with one technical modification. Consider first that for any fixed $\vec{t}'$, the problem of finding $X$ that minimizes $\mathcal{O}'_{\vec{w}}(\vec{t}', X)$ is an instance of $GOOP(\mathcal{F},\mathcal{V}, \mathcal{O})$. Simply let $w = w_O$, $f = \sum_i \sum_{t \in T_i} \frac{w_i (t, t')}{\Pr[t']} \cdot t(\cdot)$, and $g_i = t'_i(\cdot)$. 

So with black-box access to an $(\alpha, \beta)$-approximation algorithm, $G$, for $GOOP(\mathcal{F},\mathcal{V},\mathcal{O})$, let $M=(A,P)$ be the mechanism that on profile $\vec{t}'$ simply runs $G$ on input $w = w_O$, $f = \sum_i \sum_{t \in T_i} \frac{w_i (t, t')}{\Pr[t']} \cdot t(\cdot)$, $\vec{g} = \vec{t}'$. We therefore get that the mechanism $M$ satisfies the following inequality:

$$\mathbb{E}_{\vec{t}' \leftarrow \mathcal{D}'}[\beta \cdot w_O \cdot \mathcal{O}(\vec{t}', A(\vec{t}')) + \sum_i \sum_{t \in T_i} \frac{w_i(t, t')}{\Pr[t']} \cdot t(A(\vec{t}'))] \leq \alpha \mathbb{E}_{\vec{t}' \leftarrow \mathcal{D}'}[\min_{X' \in \mathcal{F}} \{w_O \cdot \mathcal{O}(\vec{t}', X') + \sum_i \sum_{t \in T_i} \frac{w_i(t, t')}{\Pr[t']} \cdot t(X'))\}]$$

By Proposition~\ref{prop:VW}, this then implies that:

$$(\beta O^M, \vec{\pi}^M, \vec{p}^M) \cdot \vec{w} \leq \alpha \min_{\vec{x} \in F(\mathcal{F},\mathcal{D}', \mathcal{O})} \{\vec{x} \cdot \vec{w}\}$$

This exactly states that $\vec{\pi}^M_0$ is an $(\alpha, \beta, O)$-approximation. It is also clear that we can compute $\vec{\pi}^M_0$ efficiently: $\mathcal{D}'$ has polynomially many profiles in its support, so we can just run $\mathcal{A}$ on every profile and see what it outputs, then take an expectation to compute the necessary quantities of the implicit form. Note that this computation is the reason we bother using $\mathcal{D}'$ at all, as we cannot compute these expectations exactly in polynomial time for $\mathcal{D}$ as the support is exponential. 

Now we state the technical modification to accommodate $w_O < 0$. Recall that for any feasible implicit form $(O, \vec{\pi}, \vec{p})$, that the implicit form $(O', \vec{\pi}, \vec{p})$ is also feasible for any $O' \geq O$. So if $w_O < 0$, simply find any feasible implicit form, then set the $O$ component to $\infty$. This yields a feasible implicit form with $\vec{\pi}_I \cdot \vec{w} = -\infty$, which is clearly an $(\alpha, \beta, O)$-approximation (in fact, it is a $(1, 1, O)$-approximation). If instead the problem has a maximization objective, we may w.l.o.g. set $O = 0$ in the implicit form we output, which means that the contribution of $\mathcal{O}$ is completely ignored. So we can use the exact same approach as the $w_O \geq 0$ case and just set $w_O = 0$. 

So let $\mathcal{A}$ be the algorithm that runs $G$ on every profile as described, and computes the implicit form of this mechanism with respect to $\mathcal{D}'$. $\mathcal{A}$ clearly terminates in the desired runtime. Finally, to implement a mechanism whose implicit form $\vec{\pi}^M_0$ matches $\mathcal{A}(\vec{w})$, simply run $G$ with the required parameters on every profile.

\end{prevproof}

\begin{prevproof}{Proposition}{prop:implement}
By Proposition~\ref{prop:solveLP}, the implicit form $\vec{\pi}_I$ output by the linear program of Figure~\ref{fig:LPBMeD} is in the convex hull of $\{\mathcal{A}^\beta_S(\vec{w}_1), \ldots, \mathcal{A}^\beta_S(\vec{w}_{d+1})\}$. Therefore, the implicit form $\vec{\pi}'_I = (O/\beta, \vec{\pi}, \vec{p})$ is in the convex hull of $\{\mathcal{A}(\vec{w}_1),\ldots, \mathcal{A}(\vec{w}_{d+1})\}$. Therefore, we can implement $\vec{\pi}'_I$ with respect to $\mathcal{D}'$ by randomly sampling a direction $\vec{w}_j$ according to the convex combination, and then implementing the corresponding $\mathcal{A}(\vec{w}_j)$. Call this mechanism $M$. By Proposition~\ref{prop:goop}, this can be done time polynomial in the desired quantities. Finally, we just need to show that the guarantees hold with the desired probability.

If our target was just an interim individually rational mechanism, it would be trivial to match the prices exactly: just charge each bidder the desired prices. But if we want an ex-post IR mechanism, we need to employ a simple reduction used in Appendix~D of~\cite{DaskalakisW12}, which causes the prices to possibly err with the rest of the implicit form. To see that all guarantees hold with the desired probability, consider that the implicit form of $M$ with respect to $\mathcal{D}$ is $\epsilon$-close to $\vec{\pi}'_I$ with the desired probability. In the event that this happens, it's obvious that the desired properties hold.
\end{prevproof}

\notshow{
The proof of Theorem~\ref{thm:objective} is essentially identical to the proof of Theorem~4 in~\cite{CaiDW13b}, except replacing the linear programming framework of Theorem~H.1 of~\cite{CaiDW13} (Theorem~\ref{thm:CDW13} here) with our new linear programming framework (Theorem~\ref{thm:LPalphabeta}). We highlight the main steps where these replacements occur below for completeness.

In~\cite{CaiDW13b}, a linear program for solving MDMDP is provided that makes use of a convex region called $F(\mathcal{F},\mathcal{D},\mathcal{O})$. Without going into formal details, this is essentially the region of all feasible (not necessarily truthful) mechanisms (technically, it is the region of implicit forms of feasible mechanisms, where the ``implicit form'' of a mechanism is a succinct description of the mechanism's behavior for different inputs). There is a special coordinate, indexed by $O$, of vectors in this region that stores the expected value of the objective $\obj$ for the corresponding mechanism, and the LP optimizes only this coordinate. Specifically the LP of~\cite{CaiDW13b} was written for maximization objectives $\obj$ so it was optimizing $O$ (the expectation of $\obj$) subject to the constraint that the mechanism lies in $F(\mathcal{F},\mathcal{D},\mathcal{O})$ (i.e. it is feasible) and is also truthful. Here, we will consider the exact same LP except that, depending on whether $\obj$ is a maximization or minimization objective, we will either maximize or minimize $O$. We refer the reader to~\cite{CaiDW13b} for a formal definition of $F(\mathcal{F},\mathcal{D},\mathcal{O})$ and a better description of our linear program. Henceforth, we denote our LP by ${\rm LP}_{{\cal F}, {\cal D}, \obj}$, to emphasize its dependence on ${\cal F}$, ${\cal D}$ and $\obj$.

The bottleneck in solving ${\rm LP}_{{\cal F}, {\cal D}, \obj}$ is obtaining a separation oracle for the region $F(\mathcal{F},\mathcal{D},\mathcal{O})$, as the truthfulness constraints are easy to check. Our goal is to extract a weird separation oracle for this region from an $(\alpha,\beta,\{O\})$-optimization algorithm for $F(\mathcal{F},\mathcal{D},\mathcal{O})$ using Theorem~\ref{thm:LPalphabeta}. Property~\ref{property:converse} of Theorem~\ref{thm:LPalphabeta} will  then guarantee that we can approximately solve our LP using the weird separation oracle, since the objective only depends on the coordinates in~$\{O\}$, in particular just variable $O$. The following proposition argues that we can get a $(\alpha,\beta,\{O\})$-optimization algorithm for $F(\mathcal{F},\mathcal{D},\mathcal{O})$ from an $(\alpha,\beta)$-approximation algorithm for 2-SADP($\mathcal{F},\mathcal{V},\mathcal{O}$). Its proof is identical to that of Corollary~13 in~\cite{CaiDW13b} so we omit  it.

\begin{proposition}\label{prop:alphabeta13}
Let $G$ be an $(\alpha,\beta)$-approximation algorithm for 2-SADP($\mathcal{F},\mathcal{V},\mathcal{O}$). Then one can obtain an $(\alpha,\beta,\{O\})$-optimization\footnote{Specifically, maximization for maximization objectives $\obj$ and minimization for minimization objectives $\obj$.} algorithm $\cal A$ for $F(\mathcal{F},\mathcal{D},\mathcal{O})$ with black box access to $G$. For any direction $\vec{w}$, $\mathcal{A}$ will output the (implicit form of the) mechanism that on every profile $g_1,\ldots,g_m$ of bidder types chooses the outcome $X \in \Delta({\cal F})$ output by $G$ on input $(f,g_1,\ldots,g_m,c)$, where $f$ depends on $\vec{w}$ and $\vec{g}$, and $c$ depends on $\vec{w}$ as in Corollary 13 of~\cite{CaiDW13b}.
\end{proposition}

We are now ready to prove our theorem.\\

\begin{prevproof}{Theorem}{thm:objective}
First, as was done in~\cite{CaiDW13b}, we do not work directly with the original distribution ${\cal D}$ over bidder types, whose support may have exponential size in the input, but an approximating distribution ${\cal D}'$, whose support has size polynomial in the input and $1/\epsilon$, for some parameter $\epsilon>0$. We omit the details of how we obtain ${\cal D}'$, as they are identical to those in~\cite{CaiDW13b}. We also omit the details of why replacing ${\cal D}$ by ${\cal D}'$ for the purposes of solving MDMDP only incurs an additive cost of $\epsilon$ in the objective and the truthfulness of the mechanism, as they are also identical to those in~\cite{CaiDW13b}. We only note that we switch from ${\cal D}$ to ${\cal D'}$  for computational considerations, namely to make sure that the runtime of our algorithms is polynomial in the input. Henceforth instead of working with region $F(\mathcal{F},\mathcal{D},\mathcal{O})$ we work with region $F(\mathcal{F},\mathcal{D}',\mathcal{O})$ defined for distribution over bidder types ${\cal D}'$.

We apply Theorem~\ref{thm:LPalphabeta} to ${\rm LP}_{{\cal F}, {\cal D}', \obj}$, taking $P=F(\mathcal{F},\mathcal{D}',\mathcal{O})$ and $Q$ to be the BIC constraints, and using Proposition~\ref{prop:alphabeta13} to get an $(\alpha,\beta,\{O\})$-optimization algorithm $\cal A$ for $P$ from an $(\alpha,\beta)$-approximation algorithm $G$ for 2-SADP($\mathcal{F},\mathcal{V},\mathcal{O}$). It follows that  an ${\alpha \over \beta}$-approximate (but a priori possibly infeasible) solution $(\frac{1}{\beta}{x}_O,\vec{x}_{-O})$ to ${\rm LP}_{{\cal F}, {\cal D}', \obj}$ can be found using Ellipsoid, in the desired running time, and given only black box access to $G$ (by Properties 3 and 4 of Theorem~\ref{thm:LPalphabeta} and Proposition~\ref{prop:alphabeta13}). Additionally, Property 2 of Theorem~\ref{thm:LPalphabeta} guarantees that $(\frac{1}{\beta}{x}_O,\vec{x}_{-O})$ lies in the convex hull of $\mathcal{A}(\vec{w}_1),\ldots,\mathcal{A}(\vec{w}_k)$, for some directions $\vec{w}_1,\ldots,\vec{w}_k$, with some multipliers $p_1,\ldots,p_k$. So the (a priori possibly infeasible) mechanism $(\frac{1}{\beta}{x}_O,\vec{x}_{-O})$ can actually be implemented by first sampling direction $\vec{w}_i$ with probability $p_i$, and then on every profile $\vec{g}$ of bidder types choosing the outcome $X \in \Delta({\cal F})$ output by $G$ on input $(f,\vec{g},c)$, where $f=f(\vec{w}_i,\vec{g})$, $c=c(\vec{w}_i)$ are as specified in Proposition~\ref{prop:alphabeta13}.
\end{prevproof}
\notshow{
\costasfnote{I removed the following because I didn't like the use of weird separation oracle in it. We have two different kinds of weird separation oracles in theorems 8 and 10 so this term is not well-defined.

\begin{corollary}\label{cor:LPCaiDW}(Corollary~11 in~\cite{CaiDW13b}) If $b$ is an upper bound on the bit complexity of the  extreme points of $F(\mathcal{F},\mathcal{D},\obj)$, then with black-box access to a weird separation oracle, $WSO$, for  $\alpha F(\mathcal{F},\mathcal{D},\obj)$, an $\alpha$-approximate solution to MDMDP can be found in time polynomial in $\sum_{i\in[m]} |T_{i}|$, $b$, and the runtime of $WSO$ on inputs with bit complexity polynomial in $\sum_{i\in[m]} |T_{i}|$, $b$.
\end{corollary}}

}}
\section{Omitted Proofs from Section~\ref{sec:makespan}}\label{app:makespanfairness}
In this section we provide a proof of Theorems~\ref{thm:alphabetamakespan},~\ref{thm:alphabetafairness}, and~\ref{thm:useful}. We begin with Theorem~\ref{thm:alphabetamakespan}. Shmoys and Tardos show that if the linear program of Figure~\ref{fig:ST} outputs a feasible fractional solution, then it can be rounded to a feasible integral solution without much loss. We will refer to this linear program as $LP(t)$ for various values of $t$.

\begin{figure}[ht]
\colorbox{MyGray}{
\begin{minipage}{\textwidth} {
\noindent\textbf{Variables:}
\begin{itemize}
\item $x_{ij}$, for all machines $i$ and jobs $j$ denoting the fractional assignment of job $j$ to machine $i$.
\item $T$, denoting the maximum of the makespan and the processing time of the largest single job used.
\end{itemize}
\textbf{Constraints:}
\begin{itemize}
\item $\sum_{i=1}^k x_{ij} = 1$, for all $j$, guaranteeing that every job is assigned.
\item $\sum_{j=1}^m p_{ij} x_{ij} \leq T$, for all $i$, guaranteeing that the makespan is at most $T$. 
\item $x_{ij} \geq 0$, for all $i, j$.
\item $x_{ij} = 0$ for all $i, j$ such that $p_{ij} > t$, guaranteeing that no single job has processing time larger than $t$.
\item $T \geq t$.
\end{itemize}
\textbf{Minimizing:}
\begin{itemize}
\item $\sum_{i, j} c_{ij}x_{ij} + T$, (almost) the makespan plus cost of the fractional solution.\\
\end{itemize}}
\end{minipage}}
\caption{$LP(t)$.}
\label{fig:ST}
\end{figure}

\begin{theorem}\label{thm:ST}(\cite{ShmoysT93})
Any feasible solution to $LP(t)$ can be rounded to a feasible integral solution in polynomial time with makespan at most $T + t$ and cost at most $C$.
\end{theorem}

With Theorem~\ref{thm:ST} in hand, we can now design a $(1, 1/2)$-approximation algorithm. Define $\hat{M}(\cdot)$ as the modified makespan of an assignment to be $\hat{M}(\vec{x}) = \max\{M(x), p_{ij} | x_{ij} > 0\}$. In other words, $\hat{M}(\vec{x})$ is the larger of the makespan and the processing time of the largest single job that is fractionally assigned. Note that for any $\vec{x} \in \{0,1\}^{km}$ that $M(\vec{x}) = \hat{M}(\vec{x})$. Now consider solving $LP(t)$ for all $km$ possible values of $t$, and let $\vec{x}^*$ denote the best solution among all feasible solutions output. The following lemma states that $\vec{x}^*$ performs better than the integral optimum.

\begin{lemma}\label{lem:LPopt}
Let $\vec{x}^*$ denote the best feasible solution output among all $km$ instances of $LP(t)$, and let $\vec{y}$ denote the integral solution minimizing makespan plus cost. Then $\hat{M}(\vec{x}^*) + C(\vec{x}^*) \leq M(\vec{y}) + C(\vec{y})$.
\end{lemma}

\begin{proof}
Some job assigned in $\vec{y}$ has the largest processing time, say it is $t$. Then $\vec{y}$ is a feasible solution to $LP(t)$, and will have value $\hat{M}(\vec{y}) + C(\vec{y})$. $\vec{x}^*$ therefore satisfies $\hat{M}(\vec{x}^*) + C(\vec{x}^*) \leq \hat{M}(\vec{y}) + C(\vec{y})$. As $\vec{y}$ is an integral solution, we have $\hat{M}(\vec{y}) = M(\vec{y})$, proving the lemma.
\end{proof}

Comining Lemma~\ref{lem:LPopt} with Theorem~\ref{thm:ST} proves Theorem~\ref{thm:alphabetamakespan}.

\begin{prevproof}{Theorem}{thm:alphabetamakespan}
Consider the algorithm that solves $LP(t)$ for all $km$ values of $t$ and outputs the fractional solution $\vec{x}^*$ that is optimal among all feasible solutions found. By Lemma~\ref{lem:LPopt}, $\vec{x}^*$ is at least as good as the optimal integral solution. By Theorem~\ref{thm:ST}, we can continue by rounding $\vec{x}^*$ in polynomial time to an integral solution $\vec{x}$ satisfying $\frac{1}{2}M(\vec{x}) + C(\vec{x}) \leq \hat{M}(\vec{x}^*) + C(\vec{x}^*) \leq M(\vec{y}) + C(\vec{y})$. 
\end{prevproof}

We next prove Theorem~\ref{thm:useful}, as it will be used in the proof of Theorem~\ref{thm:alphabetafairness}.

\begin{prevproof}{Theorem}{thm:useful}
We can break the cost of $\vec{x}$ into $C(\vec{x}) = C^+(\vec{x}) + C^-(\vec{x})$, where $C^+(\vec{x})$ denotes the portion of the cost due to jobs assigned to machines with positive cost, and $C^-(\vec{x})$ denotes the portion of the cost due to jobs assigned to machines with negative cost. As $\vec{v}$ assigns all jobs to the machine with largest positive cost, we clearly have $C^+(\vec{v}) \geq C^+(\vec{x})$ and $C^-(\vec{v}) = 0$ (but may have $F(\vec{v}) = 0$). Furthermore, as $\mathbb{E}[y_{ij}] \leq x_{ij}$ for all $i, j$, we clearly have $C^-(\vec{y}) \geq C^-(\vec{x})$ (but may have $C^+(\vec{y}) = 0$). 

So there are two cases to consider. Maybe $\beta F(\vec{x}) + C^-(\vec{x}) \geq \gamma (\beta F(\vec{x}) + C(\vec{x}))$. In this case, we clearly have $z\cdot \beta F(\vec{y}) + C(\vec{y}) \geq \beta F(\vec{x}) + C^-(\vec{x}) \geq \gamma ( \beta F(\vec{x}) + C(\vec{x}))$, and the first possibility holds. The other case is that maybe $C^+(\vec{x}) \geq (1-\gamma) (\beta F(\vec{x}) + C(\vec{x}))$, in which case we clearly have $F(\vec{v}) + C(\vec{v}) \geq C^+(\vec{x}) \geq (1-\gamma)(\beta F(\vec{x}) + C(\vec{x}))$, and the second possibility holds.
\end{prevproof}

With Theorem~\ref{thm:useful}, we may now prove Theorem~\ref{thm:alphabetafairness}. We begin by describing our algorithm modifying that of Asadpour and Siberi, which starts by solving a linear program known as the configuration LP. We modify the LP slightly to minimize fairness plus cost, but this does not affect the ability to solve this LP in polynomial time via the same approach used by Bansal and Sviridenko~\cite{BansalS06}.\footnote{Note that this is non-trivial, as the LP has exponentially-many variables. The approach of Bansal and Sviridenko is to solve the dual LP via a separation oracle which requires solving a knapsack problem.} The modified configuration LP is in Figure~\ref{fig:config}. Note that $T$ is a parameter, and for any $T$ we call the instantiation of the configuration LP $CLP(T)$. A configuration is a set $S$ of jobs. A configuration $S$ is said to be ``valid'' for machine $i$ if $\sum_{j \in S} p_{ij} \geq T$, or if $S$ contains a single job with $p_{ij} \geq T/\sqrt{k}\log^3(k)$. Call the former types of configurations ``small'' and the latter ``big.'' $S(i, T)$ denotes the set of all configurations that are valid for machine $i$.

\begin{figure}[ht]
\colorbox{MyGray}{
\begin{minipage}{\textwidth} {
\noindent\textbf{Variables:}
\begin{itemize}
\item $x_{i,S}$, for all machines $i$ and configurations $S$ denoting the fractional assignment of configuration $S$ to machine $i$.
\end{itemize}
\textbf{Constraints:}
\begin{itemize}
\item $\sum_{S \in S(i, T)} x_{i,S} = 1$, for all $i$, guaranteeing that every machine is fractionally assigned a valid configuration with weight $1$.
\item $\sum_i \sum_{S | j \in S} x_{i, S} \leq 1$, for all $j$, guaranteeing that no job is fractionally assigned with weight more than $1$.
\item $x_{i,S} \geq 0$, for all $i, C$.
\end{itemize}
\textbf{Maximizing:}
\begin{itemize}
\item $\sum_i \sum_{S \in S(i,T)} x_{i, S} \sum_{j \in S} c_{ij}$, the cost of the fractional solution $\vec{x}$.\\
\end{itemize}}
\end{minipage}}
\caption{(a modification of) The configuration LP parameterized by $T$.}
\label{fig:config}
\end{figure}

Step one of the algorithm solves $CLP(T)$ for all $T = 2^x$ for which the fairness of the optimal solution could possibly be between $2^x$ and $2^{x+1}$. It's clear that there are only polynomially many (in the bit complexity of the processing times and $k$ and $m$) such $x$. Let $\vec{x}(T)$ denote the solution found by solving $CLP(T)$ (if one was found at all). Then define $\vec{x}^* = \argmax_T \{2T + C(\vec{x}(T))\}$. We first claim that $\vec{x}^*$ is a good fractional solution.

\begin{lemma}\label{lem:stepone}
Let $OPT$ be the fairness plus cost of the optimal integral allocation. Then $2F(\vec{x}^*) + C(\vec{x}^*) \geq OPT$.
\end{lemma}

\begin{proof}
Whatever the optimal integral allocation, $\vec{z}$ is, it has some fairness $F(\vec{z})$. For $T = 2^x$ satisfying $F(\vec{z}) \in [T, 2T)$, $\vec{z}$ is clearly a feasible solution to $CLP(T)$, and therefore we must have $C(\vec{x}(T)) \geq C(\vec{z})$. As we also clearly have $2T \geq F(\vec{z})$ by choice of $T$, we necessariliy have $2F(\vec{x}(T)) + C(\vec{x}(T)) \geq OPT$. As $\vec{x}^*$ maximizes $2F(\vec{x}(T)) + C(\vec{x}(T))$ over all $T$, it satisfies the same inequality as well.
\end{proof}

From here, we will make use of Theorem~\ref{thm:useful}: either we will choose the allocation $\vec{v}$ that assigns every job to the machine with the highest non-negative cost, or we'll round $\vec{x}^*$ to $\vec{y}$ via the procedure used in~\cite{AsadpourS07}. We first state the rounding algorithm of~\cite{AsadpourS07}. 
\begin{enumerate}
\item Make a bipartite graph with $k$ nodes (one for each machine) on the left and $m$ nodes (one for each job) on the right. 
\item For each machine $i$ and job $j$, compute $x_{ij} = \sum_{S\ni j} x_{i, S}$. If $p_{ij} \geq T/\sqrt{k} \log^3 k$, put an edge of weight $x_{ij}$ between machine $i$ and job $j$. Call the resulting graph $\mathcal{M}$. 
\item For each node $v$, denote by $m_v$ the sum of weights of edges incident to $v$.
\item Update the weights in $\mathcal{M}$ to remove all cycles. This can be done without decreasing $C(\vec{x})$ or changing $m_v$ for any $v$, and is proved in Lemma~\ref{lem:cycles}.
\item Pick a random matching $M$ in $\mathcal{M}$ according to Algorithm~2 of~\cite{AsadpourS07}. Each edge $(i,j)$ will be included in $M$ with probability exactly $x_{ij}$, and each machine $i$ will be matched with probability exactly $m_i$.
\item For all machines that were unmatched in $M$, select small configuration $S$ with probability $x_{i, S}/m_i$. 
\item For all jobs that were selected both in the matching stage and the latter stage, award them just to whatever machine received them in the matching. For all jobs that were selected only in the latter stage, choose a machine uniformly at random among those who selected it. Throw away all unselected jobs.
\end{enumerate}

Before continuing, let's prove that we can efficiently remove cycles without decreasing the cost or changing any $m_v$.

\begin{lemma}\label{lem:cycles}
Let $\mathcal{M}$ be a bipartite graph with an edge of weight $x_{ij}$ between node $i$ and node $j$, and denote by $m_v$ the sum of weights of edges incident to $v$. Let also each edge have a cost, $c_{ij}$. Then we can modify the weights of $\mathcal{M}$ in poly-time so that $\mathcal{M}$ is acyclic, without decreasing $\sum_{i, j} x_{ij} c_{ij}$ or changing any $m_v$.
\end{lemma}

\begin{proof}
Consider any cycle $e_1, \ldots, e_{2x}$. For $e = (i,j)$, denote by $c(e) = c_{ij}$ and $x(e) = x_{ij}$. Call the odd edges those with odd subscripts and the even edges those with even subscripts. W.l.o.g. assume that the odd edges have higher total cost. That is, $\sum_{z=1}^x c(e_{2z-1}) \geq \sum_{z=1}^x c(e_{2z})$. Let also $\epsilon = \min_{z} x(e_{2z})$. Now consider decreasing the weight of all even edges by $\epsilon$ and increasing the weight of all odd edges by $\epsilon$. Clearly, we have not decreased the cost. It is also clear that we have not changed $m_v$ for any $v$. And finally, it is also clear that we've removed a cycle (by removing an edge). So we can repeat this procedure a polynomial number of times and result in an acyclic graph.
\end{proof}

Now, let $\vec{x}$ denote the fractional assignment obtained after removing cycles in $\mathcal{M}$. Then it's clear that $2F(\vec{x}) + C(\vec{x}) \geq OPT$. If we let $\vec{y}$ denote the randomized allocation output at the end of the procedure, it's also clear that $\mathbb{E}[y_{ij}] \leq x_{ij}$ for all $i, j$. This is because if there were never any conflicts (jobs being awarded multiple times), we would have exactly $\mathbb{E}[y_{ij}] = x_{ij}$. But because of potential conflicts, $\mathbb{E}[y_{ij}]$ can only decrease. Asadpour and Siberi show the following theorem about the quality of $\vec{y}$:

\begin{theorem}\label{thm:AS09}(\cite{AsadpourS07})
With probability $1-o(1)$, the fairness of the allocation output by the procedure is at least $F(\vec{x})/320\sqrt{k}\log^3 k$. This implies that $F(\vec{y}) \in \tilde{\Omega}(F(\vec{x})/\sqrt{k})$. 
\end{theorem}

And now we are ready to make use of Theorem~\ref{thm:useful}. 

\begin{proposition}\label{prop:AS09}
Either assigning every job to the machine with highest cost is a $(\frac{1}{2},1)$-approximation (which is a traditional $\frac{1}{2}$-approximation), or the $\vec{y}$ output by the algorithm above is a $(\frac{1}{2}, \tilde{O}(\sqrt{k}))$-approximation.
\end{proposition}
\begin{proof}
After removing cycles, we have a fractional solution $\vec{x}$ that is a $(1,2)$-approximation. By using the randomized procedure of Asadpour and Siberi, we get a randomized $\vec{y}$ satisfying $\mathbb{E}[y_{ij}] \leq x_{ij}$ for all $i,j$ and $\tilde{O}(\sqrt{k})F(\vec{y}) \geq F(\vec{x})$. Therefore, taking $\gamma = 1/2$, Theorem~\ref{thm:useful} tells us that either assigning every job to the machine with highest non-negative cost yields a $\frac{1}{2}$-approximation, or $\vec{y}$ is a $(\frac{1}{2}, \tilde{O}(\sqrt{k}))$-approximation.
\end{proof}

We conclude this section by proving that the $(m-k+1)$-approximation algorithm of Bezakova and Dani for fairness can be modified to be a $(1,m-k+1)$-approximation for fairness plus costs. The algorithm is fairly simple: for a fixed $T$, make the following bipartite graph. Put $k$ nodes on the left, one for each machine, and $m$ nodes on the right, one for each job. Put an additional $m-k$ nodes on the left for dummy machines. Put an edge from every job node $j$ to every dummy machine node of weight $\max_{i} \{0, c_{ij}\}$, and an edge from every job node to every real machine node $i$ of weight $c_{ij}$ \emph{only if} $p_{ij} \geq T$. Then find the maximum weight matching in this graph. For every job that is matched to a real machine, assign it there. For every job that is assigned to a dummy machine, assign it to the machine with the maximum non-negative cost (or nowhere if they're all negative). Call this assignment $A_T$. Denote $A^* = \argmax_T \{(m-k+1)T + C(A_T)\}$. Finally, let $V$ denote the allocation that just assigns every job to the machine of highest cost. If $F(V) + C(V) \geq (m-k+1) F(A^*) + C(A^*)$, output $V$. Otherwise, output $A^*$.

\begin{proposition}\label{prop:BD03}
The algorithm above finds an allocation $A$ satisfying $(m-k+1)F(A) + C(A) \geq OPT$. That is, $A$ is a $(1, m-k+1)$-approximation.
\end{proposition}
\begin{proof}
Consider the optimal assignment $X$. We either have $F(X) = 0$, or $F(X) > 0$. If $F(X) =0$, then clearly $X = V$. If $F(X) > 0$, then every machine is awarded at least one job, but at most $m-k+1$. For each machine $i$, define $j(i)$ to be the job assigned to $i$ with the highest processing time. Except for $j(i)$, reassign all other jobs to the machine with the highest non-negative cost. This can only increase the cost, and will not hurt the fairness by more than a factor of $m-k+1$. So this solution, $X'$, clearly has $(m-k+1)F(X') + C(X') \geq OPT$. Futhermore, $X'$ corresponds to a feasible matching when $T = F(X')$. Whatever solution $A_T$ is found instead clearly has fairness at least $T$ and cost at least $C(X')$. So $A_T$, and therefore also $A^*$, is a $(1, m-k+1)$-approximation. 

So in conclusion, either $F(X) > 0$, in which case $A^*$ is a $(1, m-k+1)$-approximation, or $F(X) = 0$, in which case $F(V) + C(V) = OPT$. So if we ever output $V$, we actually have $V = OPT$. If we output $A^*$, then either $F(X) > 0$, or $(m-k+1) F(A^*) + C(A^*) \geq F(V) +  C(V) = OPT$. In both cases, $A^*$ is a $(1,m-k+1)$-approximation.
\end{proof}

\begin{prevproof}{Theorem}{thm:alphabetafairness}
Part 1) is proved in Proposition~\ref{prop:AS09}, and part 2) is proved in Proposition~\ref{prop:BD03}.
\end{prevproof}

\end{document}